\documentclass[11pt]{article}

\usepackage[margin=1in]{geometry}
\usepackage{fullpage}
\usepackage{datetime}
\usepackage{enumitem}
\usepackage[T1]{fontenc}
\usepackage[utf8]{inputenc}
\usepackage{pgfplots}
\usepackage{amsmath, amssymb}
\usepackage{amsthm}
\usepackage{color}
\usepackage[numbers]{natbib}
\usepackage{cases}
\usepackage{xparse}
\usepackage{xargs}
\usepackage{appendix}
\usepackage{algorithm}
\usepackage{algpseudocode}
\usepackage{verbatim, xspace}
\usepackage{tikz}
\usetikzlibrary{arrows}
\usetikzlibrary{patterns}
\usetikzlibrary{calc}
\usetikzlibrary{decorations.pathreplacing,angles,quotes}
\usepackage{hyperref}

\usepackage{thmtools}
\usepackage{thm-restate}
\usepackage{cleveref}
\usepackage{xcolor}
\usepackage[colorinlistoftodos,prependcaption,textsize=tiny]{todonotes}
\newcommandx{\unsure}[2][1=]{\todo[linecolor=green,backgroundcolor=green!25,bordercolor=green,#1]{\normalsize #2}}
\newcommandx{\improvement}[2][1=]{\todo[inline,linecolor=blue,backgroundcolor=blue!05,bordercolor=blue,#1]{\normalsize #2}}
\newcommandx{\info}[2][1=]{\todo[linecolor=yellow,backgroundcolor=yellow!25,bordercolor=yellow,#1]{#2}}
\newcommandx{\floatmodel}[2][1=]{\todo[inline,linecolor=red,backgroundcolor=yellow!25,bordercolor=yellow,#1]{#2}}
\newcommandx{\thiswillnotshow}[2][1=]{\todo[disable,#1]{#2}}
\newcommandx{\karl}[2][1=]{\todo[linecolor=red,backgroundcolor=red!25,bordercolor=red,#1]{\normalsize #2}}
\newcommandx{\karolw}[2][1=]{\todo[linecolor=green,backgroundcolor=green!25,bordercolor=green,#1]{\normalsize #2}}
\newcommandx{\marvin}[2][1=]{\todo[linecolor=blue,backgroundcolor=blue!25,bordercolor=blue,#1]{\normalsize #2}}
\newcommandx{\marvinfn}[2][1=]{{\color{blue} \small #2}}

\newtheorem{theorem}{Theorem}

\newtheorem{lemma}[theorem]{Lemma}
\newtheorem{corollary}[theorem]{Corollary}
\newtheorem{claim}[theorem]{Claim}

\theoremstyle{empty}

\numberwithin{theorem}{section}
\numberwithin{lemma}{section}
\numberwithin{claim}{section}
\numberwithin{corollary}{section}
\numberwithin{definition}{section}

\newcommand{\ceil}[1]{\left\lceil #1 \right\rceil}

\newcommand{\eps}{\varepsilon}
\newcommand{\receps}{\frac{1}{\eps}}
\newcommand{\Oh}{\mathcal{O}}

\newcommand{\Otilde}{\widetilde{\mathcal{O}}}
\newcommand{\Ot}{\Otilde}
\newcommand{\tOh}{\Otilde}

\newcommand{\real}{\mathbb{R}_{+}}

\newcommand{\poly}{\mathrm{poly}}
\newcommand{\polylog}{\,\textup{polylog}}

\newcommand{\minconv}{Min-Plus Convolution\xspace}
\newcommand{\minprod}{Min-Plus Product\xspace}
\newcommand{\minmaxconv}{Min-Max Convolution\xspace}
\newcommand{\minmaxprod}{Min-Max Product\xspace}

\newcommand{\sparsity}{Tree Sparsity\xspace}

\newcommand{\strongapx}{strongly polynomial $(1+\eps)$-approximation\xspace}

\newcommand{\defproblem}[3]{
  \vspace{2mm}
  \vspace{1mm}
\noindent\fbox{
  \begin{minipage}{0.95\textwidth}
  #1 \\
  {\bf{Input:}} #2  \\
  {\bf{Task:}} #3
  \end{minipage}
  }
  \vspace{2mm}
}

\title{Approximating APSP without Scaling: \\ Equivalence of Approximate Min-Plus and Exact Min-Max}

\author{ 
    Karl Bringmann\footnote{Max Planck Institute for Informatics, Saarland
    Informatics Campus, Saarbr\"ucken, Germany, \texttt{kbringma@mpi-inf.mpg.de}}
    \and Marvin K\"unnemann\footnote{Max Planck Institute for Informatics, Saarland Informatics Campus, Saarbr\"ucken, Germany,
    \texttt{marvin@mpi-inf.mpg.de}}
    \and Karol W\k{e}grzycki\footnote{
    Institute of Informatics, University of Warsaw, Warsaw, Poland,
    \texttt{k.wegrzycki@mimuw.edu.pl}. Supported by the grants 2016/21/N/ST6/01468 and 2018/28/T/ST6/00084 of the
	Polish National Science Center and project TOTAL that has received funding from
	the European Research Council (ERC) under the European Union’s Horizon 2020
	research and innovation programme (grant agreement No 677651).}
}
\date{}

\begin{document}

\maketitle

\thispagestyle{empty}
\begin{abstract}

Zwick's $(1+\eps)$-approximation algorithm for the All Pairs Shortest Path (APSP) problem runs in time $\Ot(\frac{n^\omega}{\eps} \log{W})$, where $\omega \le 2.373$ is the exponent of matrix multiplication and $W$ denotes the largest weight. 
This can be used to approximate several graph characteristics including the diameter, radius, median, minimum-weight triangle, and minimum-weight cycle in the same time bound. 

Since Zwick's algorithm uses the scaling technique, it has a factor $\log W$ in the running time. 
In this paper, we study whether APSP and related problems admit approximation schemes avoiding the scaling technique.
That is, the number of arithmetic operations should be independent of $W$; this is called \emph{strongly polynomial}.
Our main results are as follows.
\begin{itemize}
\item We design approximation schemes in strongly polynomial time $\Oh(\frac{n^\omega}{\eps} \polylog(\frac{n}{\eps}))$ for APSP on undirected graphs as well as for the graph characteristics diameter, radius, median, minimum-weight triangle, and minimum-weight cycle on directed or undirected graphs.
\item For APSP on directed graphs we design an approximation scheme in strongly polynomial time $\Oh(n^{\frac{\omega + 3}{2}} \eps^{-1} \polylog(\frac{n}{\eps}))$. This is significantly faster than the best exact algorithm.
\item We explain why our approximation scheme for APSP on directed graphs has a worse exponent than $\omega$: 
Any improvement over our exponent $\frac{\omega + 3}{2}$ would improve the best known algorithm for \minmaxprod. In fact, we prove that approximating directed APSP and exactly computing the \minmaxprod are equivalent.
\end{itemize}
Our techniques yield a framework for approximation problems over the (min,+)-semiring that can be applied more generally. In particular, we obtain the first strongly polynomial approximation scheme for \minconv in strongly subquadratic time, and we prove an equivalence of approximate \minconv and exact \minmaxconv.

\end{abstract}
\clearpage


\setcounter{page}{1}


\section{Introduction}

\emph{Scaling} is one of the most fundamental algorithmic techniques.
On a problem involving weights from a range $\{1,\ldots,W\}$, the main idea of scaling is to consider each of the $\log W$ bits  one-by-one. Roughly speaking, in each phase we only consider the current bit, which simplifies the weighted problem to an unweighted problem. 
%
%
The scaling technique was particularly successful for graph problems (e.g.,~\cite{goldberg-algorithm,gabow-tarjan-jacm,duan-weighted-matching,orlin-assignment,gabow-network-problems}).
For instance, a scaling-based algorithm solves maximum weighted matching
in time $\Oh(m\sqrt{n}\log(nW))$~\cite{gabow-tarjan-jacm}, which was recently improved to time $\Oh(m\sqrt{n}\log{W})$~\cite{duan-weighted-matching}.

However, in some situations
scaling-based algorithms may be slower than alternative approaches, since they naturally require a factor $\log{W}$ in the running time. 
In particular, in practice weights are often given as floating-point numbers, and thus $\log W$ can easily be as large as $n$, rendering most scaling-based algorithms inferior to naive approaches. 
For this reason as well as for the genuinely theoretical interest, research on \emph{strongly polynomial} algorithms received major
attention (e.g.,~\cite{eva-tardos-strongly,
orlin-strong-mincostflow,vegh-strongly,submodular-strongly}). 
We say that an algorithm runs in strongly polynomial time if its number of arithmetic operations does not depend on $W$.
For example, the fastest strongly polynomial algorithms for maximum weight matching run in time $\Ot(nm)$~\cite{edmonds-karp,thorup}.
(Here and throughout the paper we let $\tOh(T) = \Oh(T \polylog\, T)$, in particular, $\tOh(n^c)$ never hides a factor $\log W$.)

The resulting challenge is to \emph{design improved strongly polynomial algorithms} whose running times come as close as possible to the best known scaling-based algorithms, but without any $\log W$-factors. 
%
%
%
In this paper, we tackle this challenge for a large class of approximation algorithms. 
This is achieved in part by an 
algorithmic framework that allows us to switch between approximate problems over the (min,+)-semiring and exact problems
over the (min,max)-semiring.

\subsection{Approximating APSP, Matrix Products, and Graph Characteristics}

In this paper, we study the following problems (see
Appendix~\ref{def-problems} for formal problem definitions).

\begin{itemize}[leftmargin=5.05mm]
\item \textbf{Shortest path problems:} The \emph{All-Pairs Shortest Path} problem (APSP) asks to compute, given a directed graph with positive edge weights, the length of the shortest path between any two vertices.

\item \textbf{Matrix products:} 
Given matrices $A,B \in
\real^{n\times n}$, their product over the $(\oplus,\otimes)$-semiring is the matrix $C \in \real^{n\times n}$ with $C[i,j] = \bigoplus_{1 \le k \le n}
(A[i,k] \otimes B[k,j])$. In general, the
product can be computed using $\Oh(n^3)$ semiring operations. Over the $(+,\cdot)$-ring,
the problem is standard matrix multiplication and can be solved in time $\Oh(n^\omega)
\le \Oh(n^{2.373})$~\cite{legall}.
\emph{\minprod} is the problem of computing the matrix product over the $(\min,+)$-semiring. 

\item \textbf{Graph characteristics:}
Specifically, we study the graph characteristics 
\emph{Diameter}, \emph{Radius}, \emph{Median}, \emph{Minimum-Weight Triangle}, and \emph{Minimum-Weight Cycle}.
\end{itemize} 


These graph characteristics and \minprod can be reduced to APSP, and thus all of these problems can be solved in time $\Oh(n^3)$~\cite{floyd,warshall} and using a recent algorithm by Williams~\cite{Williams18} in time $n^3 / 2^{\Omega(\sqrt{\log n})}$. Moreover, with the exception of Diameter, an $\Oh(n^{3-\delta})$-algorithm for one of these graph characteristics, or for \minprod, or for APSP would yield an $\Oh(n^{3-\delta'})$-algorithm for all of these problems~\cite{apsp-equiv,abboud-centrality}. It is therefore conjectured that none of them can be solved in truly subcubic time~\cite{apsp-equiv,abboud-centrality}.

Zwick designed a $(1+\eps)$-approximation algorithm for APSP running in time $\tOh(\frac{n^\omega}{\eps}\log{W})$~\cite{zwick-apsp}. 
This yields approximation schemes with the same guarantees for \minprod and the mentioned graph characteristics~\cite{abboud-centrality,minweightcycles}.
Zwick's running time is close to optimal\footnote{For APSP and \minprod, 
any $(1+\eps)$-approximation can be used to compute the Boolean matrix product~\cite{CohenZ01}, and thus requires time $\Omega(n^\omega)$. 
 Moreover, the dependence on $\frac 1\eps$ should be at least polynomial, since the hardness conjecture for APSP is stated for $W = \poly(n)$~\cite{apsp-equiv}, and a setting of $\eps = 1/W$ yields an exact algorithm.}, except that it is open whether the factor $\log W$ is necessary. 
To the best of our knowledge,  no strongly polynomial approximation scheme is known for any of the mentioned problems.
This leads to our main question:

\begin{center}
  \emph{Do APSP, \minprod, and the mentioned graph characteristics have strongly polynomial approximation schemes running in time~$\tOh(\frac{n^\omega}\eps)$? Or at least in time $\tOh(\frac{n^{3-\delta}}\eps)$ for some $\delta > 0$?}
\end{center}

\noindent
Note that in the setting of strongly polynomial algorithms, by \emph{time} we mean the number of arithmetic operations. However, there is also a corresponding question that considers the bit complexity. In fact, variants of our main question are reasonable and open in at least three different settings:

\begin{itemize}
\item  \emph{Number of arithmetic operations:} When we only count arithmetic operations, then in particular we can add/multiply two $\log W$-bit input integers in constant time. Thus, it is not clear why the running time of an algorithm should depend on $\log W$ at all. Nevertheless, Zwick's algorithm requires $\tOh(\frac{n^\omega}{\eps}\log{W})$ arithmetic operations. It is open whether this can be reduced to $\tOh(\frac{n^\omega}{\eps})$ (or even to $\tOh(\frac{n^{3-\delta}}\eps)$ for any $\delta > 0$).

\item  \emph{Bit complexity with integers:} In bit complexity, an arithmetic operation on $b$-bit integers has cost $\tOh(b)$. Note that the input to APSP consists of $n^2$ many $\log W$-bit integers, and suppose that we keep this number format throughout the algorithm. Running Zwick's algorithm in this setting results in a bit complexity of $\tOh(\frac{n^\omega}{\eps}\log^2{W})$, since each arithmetic operation has bit complexity $\tOh(\log(nW))$. One $\log W$-factor is natural, since we operate on $\log W$-bit integers. The question thus becomes whether the \emph{second} $\log W$-factor of Zwick's algorithm is necessary, or whether it can be improved to bit complexity $\tOh(\frac{n^\omega}{\eps}\log{W})$.

\item  \emph{Bit complexity with floating point approximations:} One can improve upon the bit complexity of Zwick's algorithm as described above by changing the number format to floating point. Note that changing any input number by a factor in $[1,1+\eps]$ changes the resulting distances by at most $1+\eps$ and thus still yields a $(1+\Oh(\eps))$-approximation. We can therefore round any input integer in the range $\{1,\ldots,W\}$ to a floating point number with an $\Oh(\log \tfrac 1\eps)$-bit mantissa and an $\Oh(\log \log W)$-bit exponent. We argue that this is the natural input format of Approximate APSP in Section~\ref{sec:machinemodel}. In this format, arithmetic operations on input numbers have bit complexity $\tOh(\log \tfrac 1\eps + \log \log W)$, and thus a factor $\log \log W$ in the bit complexity would be natural. However, implementing Zwick's algorithm in this setting yields bit complexity $\tOh(\frac{n^\omega}{\eps}\log{W})$. The question now becomes whether this can be improved to $\tOh(\frac{n^\omega}{\eps}\log \log{W})$, after converting the input numbers to floating point in time $\tOh(n^2 \log W)$ (we will ignore this conversion time throughout the paper since it is near-linear in the input size).
\end{itemize}

Note that in all three settings potentially Zwick's algorithm could be improved by a factor up to $\tOh(\log W)$.
We focus on the first setting in this paper, where our goal is to design algorithms whose number of arithmetic operations is independent of $W$. However, our algorithms also yield improvements in the other two settings, which we will briefly mention below. 

\subsection{Our Results}

In this paper, we answer our main question affirmatively for all listed problems (for Directed APSP we need the relaxed form of the question).
Our results hold on the Word RAM, see Section~\ref{sec:machinemodel} for details of the machine model.



For the mentioned graph characteristics, obtaining time $\Ot(\frac{n^\omega}{\eps})$ is an easy exercise.
Since the result is a single number, we can first compute a $\poly(n)$-approximation, round edge weights to obtain $W = \poly(n/\eps)$, and then use the $\Ot(\frac{n^\omega}{\eps}\log{W})$-time approximation scheme as a black box.


\newcounter{footnoteaobc}
\begin{restatable}[]{theorem}{Applications}
    \label{thm:applications}
    Diameter, Radius, Median, Minimum-Weight Triangle, and Minimum-Weight Cycle on directed and undirected graphs have approximation schemes in strongly polynomial time\footnote{Here, by time we mean the number of \emph{arithmetic operations} performed on a RAM machine. The \emph{bit complexity} of the algorithm is bounded by the number of arithmetic operations times $\log \log W$ (up to terms hidden by $\tOh$).} $\Ot(\frac{n^\omega}{\eps})$.
\end{restatable}
\setcounter{footnoteaobc}{\thefootnote}

For APSP restricted to \emph{undirected} graphs, we also obtain time $\Ot(\frac{n^\omega}{\eps})$. 
We augment an essentially standard scaling-based algorithm for APSP by
contracting light edges. This is more involved than our solution for graph
characteristics, and is inspired by an iterative algorithm
of~\citet{eva-tardos-strongly}. Similar edge contraction arguments have been used in the context of parallel algorithms for approximate APSP on undirected graphs~\cite{KleinS92,Cohen97}. 

\begin{restatable}[]{theorem}{UndirectedAPSP}
    \label{thm:undir-apsp}
    $(1+\eps)$-Approximate Undirected APSP is in strongly polynomial time\footnotemark[\thefootnoteaobc] $\Ot(\frac{n^\omega}{\eps})$.
\end{restatable}

For APSP on \emph{directed} graphs the ideas used above fail, since there are $n^2$ output numbers and we cannot contract directed edges.
As our most involved result in this paper, we obtain a \emph{truly subcubic} strongly polynomial approximation scheme for APSP; no such algorithm was known before.


\begin{theorem}
    \label{directed-apsp-easy-theorem}
    $(1+\eps)$-Approximate Directed APSP is in strongly polynomial time\footnotemark[\thefootnoteaobc] $\Ot(n^{\frac{\omega + 3}{2}} / \eps)$.
\end{theorem}

Our approximation scheme for (directed) APSP is, in fact, a reduction from \emph{approximate} APSP to the \emph{exact} problem \minmaxprod, i.e., the problem of computing the matrix product over the $(\min,\max)$-semiring. This problem is closely related to the All-Pairs Bottleneck Path problem.\footnote{In All-Pairs Bottleneck Path we are given a directed graph
with capacities on its edges, and want to determine for all vertices $u,v$ the capacity
of a single path for which a maximum amount of flow can
be routed from $u$ to $v$.}
\minmaxprod and All-Pairs Bottleneck Path can be solved in time $\tOh(n^{\frac{\omega + 3}{2}})$~\cite{fastest-apbp}, which is why this term appears in our approximation scheme for APSP.

Furthermore, our reduction also works in the other direction, which yields an \emph{equivalence} of approximation schemes for APSP and exact algorithms for \minmaxprod. In particular, for readers willing to believe that the best known running time for \minmaxprod is essentially optimal, this can be seen as a conditional lower bound for approximate APSP, showing that any improvements upon our approximation scheme in terms of the exponent of $n$ is unlikely.

%

\begin{restatable}[]{theorem}{EquivalenceTheorem}
  \label{thm:equ-apsp}
  For any $c \ge 2$, if one of the following statements is true, then all are:
  \begin{itemize}
    \item $(1+\eps)$-Approximate Directed APSP can be solved in strongly polynomial time $\tOh(n^c/ \poly(\eps))$,
    \item $(1+\eps)$-Approximate \minprod can be solved in strongly polynomial time $\tOh(n^c /\poly(\eps))$,
    \item exact \minmaxprod can be solved in strongly polynomial time $\tOh(n^c)$,
    \item exact All-Pairs Bottleneck Path can be solved in strongly polynomial time $\tOh(n^c)$.
  \end{itemize}
\end{restatable}

Our techniques also transfer to other problems over the (min,+)-semiring. In particular, we design the first strongly polynomial-time approximation scheme for \minconv. Analogously to Theorem~\ref{thm:equ-apsp}, we complement this result by an equivalence of approximating \minconv and exactly solving \minmaxconv.
\begin{theorem}
$(1+\eps)$-Approximate \minconv is in strongly polynomial time\footnotemark[\thefootnoteaobc]  $\Ot(n^{3/2}/\sqrt{\eps})$. Furthermore, for any $c \ge 1$, if one of the following statements is true, then both are:
  \begin{itemize}
    \item $(1+\eps)$-Approximate \minconv can be solved in strongly polynomial time $\tOh(n^c/ \poly(\eps))$,
    \item exact \minmaxprod can be solved in strongly polynomial time $\tOh(n^c)$.
  \end{itemize}

\end{theorem}

As an application, we obtain an approximation scheme with the same guarantees for the related Tree Sparsity problem.


    \begin{theorem}
$(1+\eps)$-Approximate \sparsity is in strongly polynomial time\footnotemark[\thefootnoteaobc] $\Ot(\frac{n^{3/2}}{\sqrt{\eps}})$. 
    \end{theorem}

We prove these results related to \minconv in Section~\ref{minconv} as Theorems~\ref{apx-minconv},~\ref{thm:conv-equ}, and Corollary~\ref{tree-sparsity-apx}.

\subsection{Technical Overview}

Our main technical contribution is the following Sum-to-Max-Covering, which yields a framework for reducing approximate problems over the (min,+)-semiring to exact or approximate problems over the (min,max)-semiring. 
The most intriguing results of this paper (the approximation scheme for directed APSP as well as the equivalence with \minmaxprod) are essentially immediate consequences of Sum-to-Max-Covering, see Sections~\ref{directed-apsp} and \ref{equiv}. Here we denote $[n] = \{1,\ldots,n\}$.

\begin{restatable}[Sum-to-Max-Covering]{theorem}{CoveringLemma}
    \label{covering-lemma-intro}
    Given vectors $A,B \in \real^n$ and $\eps >
    0$, in linear time in the output size we can compute vectors $A^{(1)},\ldots,A^{(s)}, B^{(1)},\ldots,B^{(s)} \in \real^n$ with $s =
    \Oh((\frac{1}{\eps} + \log{n})\log{\frac{1}{\eps}})$ such that for all $i,j \in [n]$:
    \[ A[i] + B[j] \le \min_{\ell \in [s]} \max\{A^{(\ell)}[i],B^{(\ell)}[j]\} \le (1+\eps)(A[i] + B[j]). \]
\end{restatable}

There are two main issues that make the proof of this statement non-trivial. 

For \emph{close pairs} $i,j$, meaning $\frac{A[i]}{B[j]} \in [\eps,\frac 1\eps]$, the sum $A[i] + B[j]$ and the maximum $\max\{A[i],B[j]\}$ differ significantly. It is thus necessary to change the values of the vectors $A,B$.
Roughly speaking, we handle this issue by splitting $A$ into vectors $A^{(\ell)}$ such that all entries $A^{(\ell)}[i], A^{(\ell)}[i']$ differ by either less than a factor $1+\eps$ or by more than a factor $\poly(1/\eps)$. Then we can choose $B^{(\ell)}$ such that $B^{(\ell)}[j]$ is approximately $A[i] + B[j]$ for all close pairs $i,j$. This ensures that for close pairs $\max\{A^{(\ell)}[i],B^{(\ell)}[j]\}$ is approximately $A[i] + B[j]$. For details see \emph{Close Covering} (Lemma~\ref{close-covering}).

For the \emph{distant pairs} $i,j$, with $\frac{A[i]}{B[j]} \not\in [\eps,\frac 1\eps]$, the sum $A[i] + B[j]$ and the maximum $\max\{A[i],B[j]\}$ differ by less than a factor $1+\eps$, so we do not have to change any values. However, we need to remove some entries (by setting them to $\infty$) in order to not interfere with close pairs. We show how to cover all distant pairs but no too-close pairs, via a recursive splitting into $\log n$ levels of chunks and treating boundaries between chunks by introducing several shifts of restricted areas.
For details see \emph{Distant Covering} (Lemma~\ref{distant-covering-lemma}).

\subsection{Further Related Work}
\label{related-work}

It is known that in general not every scaling-based algorithm can be made strongly polynomial, see, e.g., Hochbaum's work on the allocation problem~\cite{hochbaum-lower-bound}.

\paragraph{APSP and \minprod}
For undirected graphs with weights in $\{-W,\ldots,W\}$, APSP can be solved exactly in time $\Ot(W n^\omega)$~\cite{seidel,galil-margalit-apsp,alon-apsp,ShoshanZ99}, where $\omega <
2.373$ is the matrix multiplication exponent~\cite{legall}.
For directed graphs with weights in $\{-W,\ldots,W\}$, \citet{zwick-apsp} presented an
$\Oh(W^{0.68}n^{2.575})$-time algorithm that also uses fast matrix multiplication (in fact, recent advances for rectangular matrix multiplication yield slightly stronger bounds~\cite{LeGall12, LeGallU18}). 

The closest related work to our paper is by \citet{max-weight-triangle}, who considered the real-valued \minprod. They proposed a method to compute the $k$ most significant bits
of each entry of the \minprod in time $\Oh(2^k n^{2.687} \log{n})$, in the traditional comparison-addition model of computation. 
This is similar to an additive $W/2^k$-approximation.
However, it is incomparable to a $(1+\eps)$-approximation algorithm for \minprod, since (1) the $k$ most significant bits might all be 0, in which case they do not provide a multiplicative approximation, and (2) a $(1+\eps)$-approximation not necessarily allows to determine any particular bit of the result, e.g., if a number is very close to being a power of 2.
Subsequently, their dependence on $n$ was improved to $\Oh(2^k
n^{2.684})$\cite{yuster-dominance}, which was further refined to $\Oh(2^{0.96k}n^{2.687})$ and to $\Oh(2^{ck}n^{2.684})$ for some $c < 1$~\cite{legall-apsp}. 

For approximate APSP for real-valued
graphs with weights in $[-n^{o(1)},n^{o(1)}]$, \citet{yuster-apsp-real} presented an additive $\eps$-approximation in time $\Ot(n^{\frac{\omega + 3}{2}})$. 
More recently, among other results, \citet{roditty-apsp-additive} 
gave an algorithm computing every distance $d_G(u,v)$ up to an additive error of $d_G(u,v)^p$ in time $\tOh(W n^{2.575 - p/(7.4 - 2.3p)})$. For very small $W$, this interpolates between Zwick's fastest exact algorithm and his approximation algorithm~\cite{zwick-apsp}.

In this paper we will focus on the problem of $(1+\eps)$-approximating APSP when
$\eps$ is close to~$0$. For $\eps < 1$ the problem is at least as hard as Boolean
matrix multiplication~\cite{CohenZ01} and thus requires time $\Omega(n^\omega)$. However, there are more efficient
algorithms in the regime $\eps \ge 1$ for undirected graphs, using \emph{spanners} and \emph{distance oracles}~\cite{distance-oracles}.



\paragraph{All-Pairs Bottleneck Path and \minmaxprod}
The \emph{All-Pairs Bottleneck Path} (APBP) problem is, given an edge-weighted directed graph $G$, to determine for all vertices $u,v$ the maximal weight $w$ such that there is a path from $u$ to $v$ using only edges of weight at least $w$.
It is known that APBP is equivalent to \minmaxprod, up to lower order factors in running time.
The first truly subcubic algorithm for \minmaxprod was given by \citet{first-apbp}, which was improved to time $\Oh(n^{\frac{\omega+3}{2}})$ by \citet{fastest-apbp}.


\citet{apbp-vertex-weight} proposed an $\Oh(n^{2.575})$-time algorithm for a \emph{vertex-weighted} variant of APBP. 
\citet{apsp-af} introduced the problem \emph{All-Pairs Shortest Path for All
Flows} (APSP-AF) and provided an approximation algorithm in time $\Ot(n^{\frac{\omega + 3}{2}} \eps^{-3/2} \log{W})$. They also proved an equivalence with \minmaxprod. However, in contrast to the equivalences presented in this paper, their equivalence loses a factor $\log W$, and thus does not work for strongly polynomial algorithms.

APSP and APBP can be easily computed in time $\Oh(n^{2.5})$ on quantum
computers~\cite{quantum-virginia}. \citet{legall-apsp} designed the first quantum algorithm for computing
\minmaxprod in time $\Oh(n^{2.473})$, and noted that every problem
equivalent to APBP admits a nontrivial $\Oh(n^{2.5 - \eps})$-time algorithm in the
quantum realm.

It is also worth mentioning that there are efficient algorithms for products in
other algebraic structures, e.g., dominance product, $(+,\min)$-product,
$(\min,\le)$-product (see, e.g.,~\cite{vassilevska-thesis}).

\paragraph{Hardness of Approximation in P}
There is a growing literature on hardness of approximation in P (see,
e.g.,~\cite{pcp-in-p,rubinstein,AbboudR18,ChalermsookCKLM17,pcp-in-p-2,Chen18}), building on recent progress
in fine-grained complexity theory. For readers that are willing to believe that
the current algorithms for \minmaxprod are close to optimal, our equivalence of
approximating APSP and exactly computing \minmaxprod is a hardness of
approximation result, and in fact it is one of the first tight lower bounds for
approximation algorithms for problems in P (cf.~\cite{Chen18}).

\subsection{Organization}

After preliminaries on the machine model in Section~\ref{prelim}, we present our approximation scheme for APSP in Section~\ref{directed-apsp} and the equivalence with \minmaxprod in Section~\ref{equiv}. The main technical result, Max-to-Sum-Covering, is proved in Section~\ref{sec:covering-lemma}. In Section~\ref{strong-apsp} we discuss Undirected APSP, and in Section~\ref{applications} we discuss certain graph characteristics. Finally, in
Section~\ref{minconv} we present our approximation scheme for \minconv and prove the
equivalence with \minmaxconv. Formal problem definitions can be found
in Appendix~\ref{def-problems}.


\section{Preliminaries}
\label{prelim}

\paragraph{Notation}

By $W$ we denote the largest input weight. We
use $\Ot$-notation to suppress polylogarithmic factors in $n$ and $\eps$,
but never in $W$. 
By $\omega <
2.373$ we denote the exponent of matrix
multiplication~\cite{legall}.  Observe that $\frac{\omega + 3}{2}
< 2.687$. We use $[n]$ to denote the set $\{1,2,\ldots,n\}$. We assume all input weights to be positive real numbers --  in particular after scaling we can assume all numbers to be at least 1, so the input range is $\real = \{ a \ge 1 \; | \; a
\in \mathbb{R}\} \cup \{\infty\}$. 
We denote by $W$ the largest finite input number.

We will state our results for both directed and undirected graphs. By default,
$G$ denotes a graph, $V$ the set of its vertices and $E$ set of edges. In most
cases the graph is weighed with a function $w : E \rightarrow \real$. When we
talk about graph algorithms, $n$ denotes the number of vertices and $m$
the number of edges. We consider multiplicative $(1+\eps)$-approximation algorithms, where the deviation from the exact value is always one-sided. 
We assume
that $\eps > 0$ is sufficiently small ($\eps < 1/10$). For formal
definitions of the problems considered in this paper, see Appendix~\ref{def-problems}.

\subsection{Machine Model and Input Format}
\label{sec:machinemodel}

Throughout the paper, we will assume that \textbf{for all approximate problems input numbers are represented in floating-point}, while \textbf{for all exact problems input numbers are integers represented in usual bit representation}. This choice of representation is not necessary for our new approximation algorithms (they would also work on the Word RAM with input in bit representation or on the Real RAM allowing only additions and comparisons); however, it is necessary for our equivalences between approximate and exact problems, as we discuss at the end of this section. We first describe the details of these formats as well as why this choice is well-motivated and natural.

The reader is invited to skip over the machine model details and consider an
unrealistic, but significantly simpler model of computation throughout the
paper: A Real RAM model where all logical and arithmetic operations on real
numbers have unit cost, including rounding operations. This model is too
powerful to be a realistic model of computation~\cite{real-ram-power}, but considering our algorithms in this model captures the main ideas.

\paragraph{Floating-Point Representation for Approximate Problems}
For all approximate problems considered in this paper, we can change every input weight by a factor $1+\eps$ in a preprocessing step; this changes the result by at most a factor $1+\eps$. It therefore suffices to store for each input weight $w$ its rounded logarithm $e = \lfloor \log_2 w \rfloor$, which requires only $\Oh(\log \log W)$ bits, and a $(1+\eps)$-approximation of $w/2^e \in [1,2]$, which requires only $\Oh(\log 1/\eps)$ bits. Note that this is floating-point representation. Hence, floating-point is the natural input format for the approximate problems studied in this paper!

The necessity for rigorous models for floating-point numbers in theoretical computer science was observed in~\cite{floating-point-1,floating-point-2,floating-point-4}.
Here we follow the format proposed by \citet{thorup}, except that we slightly simplify it, since we only want to represent positive reals. 
In floating-point representation, a positive real number is given as a pair $x = (e,f)$, where the \emph{exponent} $e$ is a $\kappa$-bit
integer and the \emph{mantissa} $f$ is a $\gamma$-bit string $f_1,\ldots,f_\gamma$. The pair $x$ represents the real number
\begin{displaymath}
    2^e \cdot \Big(1 + \sum_{i=1}^\gamma f_i/2^i\Big).
\end{displaymath}
Here $\gamma,\kappa$ are parameters of the model. Moreover, we assume that all arithmetic operations on floating-point numbers can be performed in constant time.


For all approximate problems considered in this paper, we assume the input weights to be given in floating-point format. In particular, if the input weights are in the range $[1,W]$, we assume floating-point representation with $\Theta(\log n)$-bit mantissa and $\Theta(\log n + \log \log W)$-bit exponent. 
The unit-cost assumption (that all arithmetic operations on floating-point numbers take constant time) thus hides at most a factor $\tOh(\log n + \log \log W)$ compared to, e.g., the complexity of performing these operations by a device operating on bits. Note that many other formats can be efficiently converted into floating-point, and thus our algorithms also work in other settings.

Note that using a fixed floating-point precision introduces inherent inaccuracies when performing arithmetic operations. For simplicity of presentation, however, we shall assume that all arithmetic operations yield an exact result. For the algorithms in this paper, it is easy to see that this assumption can be removed by increasing the precision slightly.

\paragraph{Bit Representation for Exact Problems}
The only two exact problems that we consider in this problem are \minmaxprod and \minmaxconv. 
Since both problems are of the Min-Max type, it is easy to see that we can replace all input numbers by their \emph{ranks}, i.e., their index in the sorted ordering of all input numbers. Solving the problem on the ranks, we can then infer the result. Hence, up to additional near-linear time in the input size to determine the ranks, we can assume that all input numbers are integers in the range $\{1,\ldots,\textup{poly}(n)\}$, and thus all input numbers are $\Oh(\log n)$-bit integers. This is the reason why for the exact problems studied in this paper, bit representation is the natural input format, and not floating-point!
As usual for the Word RAM, we assume that each memory cell stores $\Omega(\log n)$-bit integers, and thus operations on input numbers can be performed in constant time.

\paragraph{Necessity of our Choice of Input Representation}
We crucially use our choice of input formats in our equivalences of approximate
Min-Plus and exact Min-Max problems (see Theorem~\ref{thm:equ-apsp}, Theorem~\ref{thm:conv-equ}):
In the reduction from exact Min-Max to approximate Min-Plus we need to exponentiate some numbers. In usual bit representation, this would translate $\Oh(\log n)$-bit integers to $\poly(n)$-bit integers and thus not be efficient enough. However, if $m$ is an $\Oh(\log n)$-bit integer in standard bit-representation, then we can store $2^m$ in floating-point representation by storing $m$ as the exponent; the resulting floating-point number has an $\Oh(\log n)$-bit exponent (and an $\Oh(1)$-bit mantissa). 

For the other direction, from approximate problems in floating-point to exact
problems in bit representation, we use that for Min-Max problems we can replace input numbers by their ranks, which converts floating-point numbers to $\Oh(\log n)$-bit integers in bit representation.


\section{Strongly Polynomial Approximation for Directed APSP}
\label{directed-apsp}

We present a strongly polynomial $(1+\eps)$-approximation algorithm for APSP with running time $\tOh(n^{\frac{\omega + 3}{2}} \eps^{-1})$, proving Theorem~\ref{directed-apsp-easy-theorem}. 
To this end, we first recall the reduction from approximate APSP to approximate \minprod
from~\cite{zwick-apsp} (see Theorem~\ref{thm:red-apsp-minprod}). Then we observe that Sum-To-Max-Covering yields a reduction from approximate \minprod to \minmaxprod. Using the known $\tOh(n^{\frac{\omega + 3}{2}})$-time algorithm for the latter shows the result (see Theorem~\ref{thm:apx-minprod}).



\begin{theorem}[Implicit in \cite{zwick-apsp}]
    \label{thm:red-apsp-minprod}
    If $(1+\eps)$-Approximate \minprod can be solved in time $T(n,\eps)$, then $(1+\eps)$-Approximate APSP can be solved in time $\Oh\big( T(n,\eps/\log n) \cdot \log n\big)$.
\end{theorem}


\begin{proof}
    For the sake of completeness, we repeat the argument of \citet[Theorem 8.1]{zwick-apsp}. 
    Let $A$ be the adjacency matrix of a given edge-weighted directed graph $G$, i.e., if there is an edge
    $(i,j) \in E$ of weight $w(i,j)$ then $A[i,j] = w(i,j)$, and $A[i,j] = \infty$ otherwise. We also add self-loops of weight~0, i.e, we set $A[i,i] = 0$ for all $i \in [n]$. Given $\eps > 0$, we set $\eps' := \ln(1+\eps) / \ceil{\log{n}}$ (where $\ln$ is the natural logarithm and $\log$ is base 2). We will
    perform $\ceil{\log{n}}$ iterations of repeated squaring. In
    each iteration, we execute $(1+\eps')$-Approximate \minprod on the current matrix $A$ with itself, i.e., we square the current matrix $A$. An easy inductive proof shows that after $r$ iterations each entry $A[i,j]$ is bounded from below by the distance from $i$ to $j$ in $G$, and bounded from above by $(1+\eps')^r$ times the length of the shortest $2^r$-hop path from $i$ to $j$. Since any shortest path uses at most $n$ edges, after $\lceil \log n \rceil$ iterations each entry $A[i,j]$ is an approximation of the distance from $i$ to $j$ in $G$, by a multiplicative factor of 
    \[ (1+\eps')^{\lceil \log n \rceil} = \Big(1+ \frac{\ln(1+\eps)}{\ceil{\log{n}}}\Big)^{\lceil \log n \rceil}
    \le 1+\eps.
    \]
    The direct running time of the reduction is $\Oh(n^2\log n)$ and there are
    $\Oh(\log{n})$ calls to $(1+\eps')$-Approximate \minprod with $\eps' = \Theta(\frac{\eps}{\log{n}})$.
\end{proof}


\begin{theorem}
    \label{thm:apx-minprod}
    $(1+\eps)$-Approximate \minprod can be solved in time $\tOh(n^{\frac{\omega + 3}{2}} \eps^{-1})$.
\end{theorem}

\begin{proof}
  We use Sum-To-Max-Covering to reduce approximate \minprod to exact \minmaxprod and then use a known algorithm for the latter; the pseudocode is shown in Algorithm~\ref{alg:approxi-minprod-via-maxprod}. 

    Consider input matrices $A,B\in \real^{n\times n}$ on which we want to compute
    $C \in \real^{n\times n}$ with $C[i,j] = \min_{k \in [n]} \{ A[i,k] + B[k,j]\}$
    for all $i,j \in [n]$.  We view the matrices $A,B$ as vectors in
    $\real^{n^2}$, in order to apply Sum-To-Max-Covering
    (Lemma~\ref{covering-lemma}). This yields vectors
    $A^{(1)},\ldots,A^{(s)},\, B^{(1)},\ldots,B^{(s)} \in \real^{n^2}$, which we re-interpret as matrices in $\real^{n\times n}$. We compute the \minmaxprod of
    every layer $A^{(\ell)},B^{(\ell)}$ and return the entry-wise minimum of the
    results, see Algorithm~\ref{alg:approxi-minprod-via-maxprod}. (Note that we
    can replace the entries of $A^{(\ell)},B^{(\ell)}$ by their ranks before
    computing the \minmaxprod and then infer the actual result --- this is
    necessary since our input format for approximate \minprod is floating-point,
    but for \minmaxprod our input format is standard bit representation.)

\begin{algorithm}
    \caption{$\textsc{ApproximateMinProd}(A,B,\eps)$.}
	\label{alg:approxi-minprod-via-maxprod}
\begin{algorithmic}[1]
    \State $\{ (A^{(1)},B^{(1)}),\ldots,(A^{(s)},B^{(s)})\} = \textsc{SumToMaxCovering}(A,B,\eps)$
    \State $C^{(\ell)} := \textsc{MinMaxProd}(A^{(\ell)}, B^{(\ell)})$ for all $\ell \in [s]$
    \State $\tilde C[i,j] := \min_{\ell \in [s]} C^{(\ell)}[i,j]$ for all $i,j \in [n]$
    \State \Return $\tilde C$
\end{algorithmic}
\end{algorithm}

   Let us prove that the output matrix $\tilde C$ is a
    $(1+\eps)$-approximation of $C$. 
    Sum-To-Max-Covering yields that for any $i,j,k$ we have
    \begin{displaymath} A[i,k] + B[k,j]  \le \min_{\ell \in [s]} \max\{A^{(\ell)}[i,k],B^{(\ell)}[k,j]\} \le (1+\eps)(A[i,k] + B[k,j]).
    \end{displaymath}
    In particular, since $C[i,j] = \min_{\ell \in [s]} C^{(\ell)}[i,j] =
    \min_{\ell \in [s]} \min_{k \in [n]}
    \max\{A^{(\ell)}[i,k],B^{(\ell)}[k,j]\},$ and $C[i,j] = \min_{k \in [n]} (A[i,k] + B[k,j])$, we obtain
    \[ C[i,j] \le \tilde C[i,j] \le (1+\eps) C[i,j]. \]
    
    

    Sum-To-Max-Covering runs in time $\Ot(n^2/\eps)$. Computing $s$ times the
    \minmaxprod runs in time $\tOh(s n^{\frac{\omega + 3}{2}})$. We conclude the proof by noting that Sum-To-Max-Covering yields $s =
    \Oh(\frac 1\eps \polylog(n/\eps))$. 
\end{proof}

Combining Theorems~\ref{thm:red-apsp-minprod} and \ref{thm:apx-minprod} yields a $(1+\eps)$-approximation for APSP in time $\tOh(n^{\frac{\omega + 3}{2}} \eps^{-1})$. 

\section{Equivalence of Approximate APSP and Min-Max Product}
\label{equiv}

We next prove our equivalence of approximating APSP, exactly computing the \minmaxprod, and other problems. The theorem is restated here for convenience.

\EquivalenceTheorem*


\begin{proof}
  Equivalence of $(1+\eps)$-Approximate APSP and $(1+\eps)$-Approximate \minprod is essentially known. One direction is given by Theorem~\ref{thm:red-apsp-minprod}. For the other direction, given matrices $A,B$ we build a 3-layered graph, with edge weights between the first two layers as in $A$, edge weights between the last two layers as in $B$, and all edges directed from left to right. Then we observe that the pairwise distances between the first and third layers are in one-to-one correspondence to \minprod on $A,B$, also in an approximate setting. 
  
  Equivalence of \minmaxprod and All-Pairs Bottleneck Path is folklore (see, e.g.,~\cite{fastest-apbp}). Both directions of this equivalence work exactly as for (approximate) \minprod vs.\ APSP.
  
  Our main contribution is the equivalence of $(1+\eps)$-Approximate \minprod and exact \minmaxprod. 
  Observe that if \minmaxprod can be solved in time $T(n)$ then the algorithm from Theorem~\ref{thm:apx-minprod} runs in time $\tOh(T(n)/\eps)$. 
  
  It remains to show a reduction from \minmaxprod to $(1+\eps)$-Approximate \minprod. Fix any constant $\eps > 0$.
  Given matrices $A,B$, denote their \minmaxprod by $C$. Let $r$ be the value of $4
  (1+\eps)^2$ rounded up to the next power of 2, and consider the matrices
  $A',B'$ with $A'[i,j] := r^{A[i,j]}$ and $B'[i,j] := r^{B[i,j]}$. (Recall that
  the input $A,B$ for \minmaxprod is in standard bit representation, so in
  constant time we can compute $r^{A[i,j]}$ in floating-point representation, by
  writing $A[i,j] \cdot \log r$ into the exponent.) Let~$C'$ be the result of
  $(1+\eps)$-Approximate \minprod on $A',B'$. 
  \begin{claim} \label{cla:red-minmaxprod-minprod}
    We have $r^{C[i,j]} \le C'[i,j] \le r^{C[i,j] + 1/2}$ for all $i,j$.
  \end{claim}  
  Using this claim, we can infer $C$ from $C'$ by computing $C[i,j] = \lfloor \log_r C'[i,j] \rfloor$ (i.e., we simply read the most significant bits of the exponent of the floating-point number $C'[i,j]$). 
  If $(1+\eps)$-Approximate \minprod  can be solved in time $T(n)$ (recall that $\eps$ is fixed), then this yields an algorithm for \minmaxprod running in time~$\tOh(T(n))$.
  \end{proof}
  \begin{proof}[Proof of Claim~\ref{cla:red-minmaxprod-minprod}]
    We will use $\min_{k}(A'[i,k]+B'[k,j]) \le C'[i,j] \le
    (1+\eps)\min_{k}(A'[i,k]+B'[k,j])$ for all $i,j \in [n]$.
    For any $i,j$ there exists $k$ with $C[i,j] = \max\{A[i,k],B[k,j]\}$. Hence,
  \begin{displaymath}
    C'[i,j] \le (1+\eps)(A'[i,k] + B'[k,j]) = (1+\eps)(r^{A[i,k]} + r^{B[k,j]})
    \le 2(1+\eps)r^{\max\{A[i,k],B[k,j]\}} = 2(1+\eps) r^{C[i,j]},
  \end{displaymath}
  and by $r \ge 4 (1+\eps)^2$ we obtain $C'[i,j] \le r^{C[i,j]+1/2}$.
  Moreover, for any $i,j$ there exists $k$ with $C'[i,j] \ge A'[i,k]+B'[k,j]$. We thus obtain
  \begin{displaymath} 
  C'[i,j] \ge A'[i,k] + B'[k,j] = r^{A[i,k]} + r^{B[k,j]}
  \ge r^{\max\{A[i,k], B[k,j]\}} \ge r^{C[i,j]}. \qedhere 
  \end{displaymath}
  \end{proof}

We remark that for scaling algorithms this proof shows an equivalence of 
the $\Ot(W n^\omega)$-time exact algorithm for \minmaxprod and the $\Ot(\frac{n^\omega}{\poly(\eps)} \log{W})$-time approximation scheme for \minprod.

\section{Sum-To-Max-Covering}
\label{sec:covering-lemma}

In this section, we prove the main technical result of this paper, which we slightly reformulate here.

\begin{theorem}[Sum-to-Max-Covering, Reformulated]
    \label{covering-lemma}
    Given vectors $A,B \in \real^n$ and a parameter $\eps > 0$,
     there are vectors $A^{(1)},\ldots,A^{(s)}, B^{(1)},\ldots,B^{(s)} \in
     \real^n$ with $s = \Oh(\frac{1}{\eps}\log{\frac{1}{\eps}} + \log{n}\log{\frac{1}{\eps}})$ and:
    \begin{enumerate}[align=left, font=\normalfont, label=(\roman*)]
    	 \item for all $i,j \in [n]$ and all $\ell \in [s]$: \[\max\{A^{(\ell)}[i],B^{(\ell)}[j]\} \ge A[i] + B[j],\]
    	 \item for all $i,j \in [n]$ there exists $\ell \in [s]$: \[\max\{A^{(\ell)}[i],B^{(\ell)}[j]\} \le (1+\eps)(A[i] + B[j]).\]
    \end{enumerate}
    We can compute such vectors $A^{(1)},\ldots,A^{(s)}, B^{(1)},\ldots,B^{(s)}$ in time $\Oh(\frac{n}{\eps}\log{\frac{1}{\eps}} +
    n\log{n}\log{\frac{1}{\eps}})$.
\end{theorem}


We split the construction into two parts, covering the pairs $i,j$ with $\frac{A[i]}{B[j]} \in [\eps,1/\eps]$ (Close Covering Lemma, Section~\ref{sec:close-covering}) and covering the remaining pairs (Distant Covering Lemma, Section~\ref{distant-covering}). We show how to combine both cases in Section~\ref{combining-covering}.

\subsection{Close Covering}
\label{sec:close-covering}

We first want to cover all pairs $i,j$ with $\frac{A[i]}{B[j]} \in [\eps,1/\eps]$. 
To get an intuition, let $d \in \mathbb{Z}$ and consider only the entries $A[i]$ in the range $[(1+\eps)^{d-1}, (1+\eps)^d)$. Remove all other entries of $A$ by setting them to $\infty$, obtaining a vector $A'$. Since we consider the close case, we are only interested in entries $B[j]$ that differ by at most a factor $1/\eps$ from $A[i]$, so consider the entries $B[j]$ in the range $[\eps (1+\eps)^{d-1}, \frac 1\eps (1+\eps)^d)$. Add $(1+\eps)^d$ to all such entries $B[j]$ and remove all other entries of $B$ by setting them to $\infty$, obtaining a vector $B'$.
Then for the considered entries we have $\max\{A'[i],B'[j]\} = B'[j] = B[j] + (1+\eps)^d$, which is between $A[i]+B[j]$ and $(1+\eps)(A[i] + B[j])$. 
This covers all considered pairs in the sense of Max-to-Sum-Covering. 

However, naively we would need to repeat this construction for too many values of $d$. 
The main observation of our construction is that we can perform this construction in parallel for all values $d \in D = \{s, 2s, 3s, \ldots\}$. That is, we only remove an entry of $A$ if it is irrelevant for \emph{all} $d \in D$, and similarly for the entries of $B$. For a sufficiently large integer $s = \Theta(\frac 1\eps \log \frac 1\eps)$, it turns out that the considered entries for different $d$'s do not interfere. Performing this construction for all shifts $D+1, D+2, \ldots, D+s$ covers all close pairs.
See Figure~\ref{fig:apx-minconv-q} for an illustration.


%

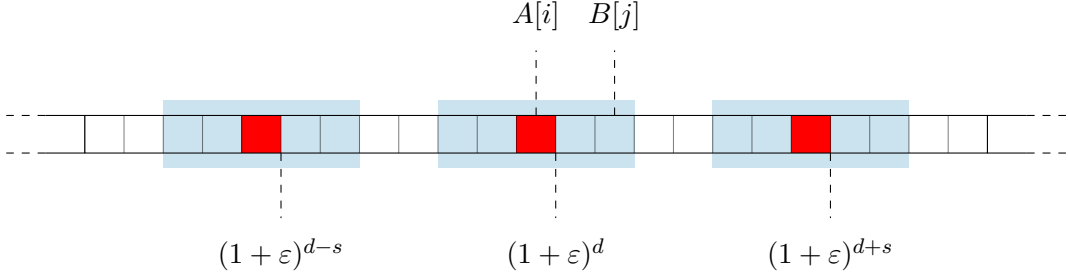
\begin{figure}[ht!]
    \centering
   \begin{tikzpicture}[]

       \def\widthrectangle{12.0}
       \def\heightrectangle{0.5}
       \def\numberoflines{23}
       \def\highlight{0.2} 
       \def\highlightlength{5}
       \def\bluecolor{black!30!cyan!20}
       \def\greencolor{red}

       \def\unit{\widthrectangle/\numberoflines}

       \newcommand{\drawhighlight}[1]{
           \fill[fill=\bluecolor] (#1*\unit,-\highlight) rectangle (#1*\unit + \highlightlength*\unit,\heightrectangle+\highlight);
       }
       \newcommand{\drawgreen}[1]{
           \fill[fill=\greencolor] (#1*\unit,0) rectangle (#1*\unit + \unit,\heightrectangle);
       }

       \newcommand{\drawdashed}[3]{
           \node (end) at (#1,#2 + #3) {};
           \draw[dashed] (#1,#2) -- (end);
       }
       \newcommand{\lastlabel}[2]{
           \node at (end) [#2] {#1};
       }

       \foreach \x in {2,9,16}
             \drawhighlight{\x};
       \foreach \x in {4,11,18}
            \drawgreen{\x};

       \foreach \x in {0,...,\numberoflines}
         \draw[opacity=0.5] (\x*\unit,0) -- (\x*\unit, \heightrectangle);
       \drawdashed{\unit*5}{0}{-1};
       \lastlabel{$(1+\eps)^{d-s}$}{below};
       \drawdashed{\unit*12}{0}{-1};
       \lastlabel{$(1+\eps)^{d}$}{below};
       \drawdashed{\unit*19}{0}{-1};
       \lastlabel{$(1+\eps)^{d+s}$}{below};

       \node (end_rectangle) at (\widthrectangle,\heightrectangle) {};

       \draw (0,0) rectangle (end_rectangle);
       \drawdashed{\unit*11.5}{\heightrectangle}{1};
       \lastlabel{$A[i]$}{above};
       \drawdashed{\unit*13.5}{\heightrectangle}{1};
       \lastlabel{$B[j]$}{above};


       \draw (0,0) -- (-\unit,0);
       \draw [dashed] (-\unit,0) -- (-2*\unit,0);
       \draw (0,\heightrectangle) -- (-1*\unit,\heightrectangle);
       \draw [dashed] (-\unit,\heightrectangle) -- (-2*\unit,\heightrectangle);


       \draw (\widthrectangle,0) -- (\widthrectangle+\unit,0);
       \draw [dashed] (\widthrectangle+\unit,0) -- (\widthrectangle+2*\unit,0);
       \draw (\widthrectangle,\heightrectangle) -- (\widthrectangle+\unit,\heightrectangle);
       \draw [dashed] (\widthrectangle+\unit,\heightrectangle) -- (\widthrectangle + 2*\unit,\heightrectangle);

   \end{tikzpicture}
    \caption{An illustration for Algorithm~\ref{alg:close-covering}. Shown is the positive real line in log-scale.
    Entries of $A$ that lie outside the red/dark-shaded areas are set to $\infty$. Entries of $B$ that lie outside the blue/light-shaded areas are set to $\infty$. We set $A^{(\ell)}[i] := (1+\eps)^d$ and $B^{(\ell)}[j] := B[j] + (1+\eps)^d$. This guarantees the approximation $(1-\eps)(A[i] + B[j]) \le \max\{A^{(\ell)}[i],B^{(\ell)}[j]\} \le (1+\eps)(A[i] + B[j])$ for close pairs. 
    Numbers in non-overlapping parts differ by so much that their sum and their max are equal up to a factor $1+\eps$. This ensures that they do not interfere with the close pairs.
    }
   \label{fig:apx-minconv-q}
\end{figure}


\begin{lemma}[Close Covering]
    \label{close-covering}
    Given vectors $A,B \in \real^n$ and a parameter $\eps >
    0$, there are vectors $A^{(1)},\ldots,A^{(s)}, B^{(1)},\ldots,B^{(s)} \in \real^n$ with $s =
    \Oh(\frac{1}{\eps}\log \frac{1}{\eps})$ such that:
    \begin{enumerate}[align=left, font=\normalfont, label=(\roman*)]
        \item for all $i,j \in [n]$ and all $\ell \in [s]\colon \; \max\{A^{(\ell)}[i],B^{(\ell)}[j]\} \ge (1-\eps)(A[i] + B[j])$, and
        \item for all $i,j \in [n]$ if $\frac{A[i]}{B[j]} \in [\eps,1/\eps]$ then $\exists \ell \in [s]\colon \; \max\{A^{(\ell)}[i],B^{(\ell)}[j]\} \le (1+\eps)(A[i] + B[j])$.
    \end{enumerate}
    We can compute such vectors $A^{(1)},\ldots,A^{(s)}, B^{(1)},\ldots,B^{(s)}$ in time $\Oh(\frac{n}{\eps} \log \frac{1}{\eps})$.
\end{lemma}
\begin{proof}
We choose $s = \Theta(\frac 1\eps \log \frac 1\eps)$ with sufficiently large hidden constant, and for any $\ell \in \{1,\ldots,s\}$ construct vectors $A^{(\ell)}, B^{(\ell)}$ as described in Algorithm~\ref{alg:close-covering}.

\begin{algorithm}
    \caption{$\textsc{CloseCovering}(A,B,\eps)$.}
	\label{alg:close-covering}
\begin{algorithmic}[1]
    \State Set $s := 1 + \lceil 2 \log_{1+\eps} (1/\eps)\rceil$ 
    \Comment{Note that $s = \Theta(\frac 1\eps \log \frac 1\eps)$.}
    \For {$\ell=1,\ldots,s$}
        \Comment{Take care of $A[i] \approx (1+\eps)^{k s + \ell}$ for any $k$.}
        \State $D_\ell := \{ k s + \ell \; | \; k \in \mathbb{Z} \}$
        \State $A^{(\ell)}[i] := \begin{cases} (1+\eps)^d & \text{if} \; A[i] \in \big[(1+\eps)^{d-1}, (1+\eps)^d\big) \; \text{for some} \; d \in D_\ell \\
                                        \infty & \text{otherwise}
        \end{cases}$
        \State $B^{(\ell)}[j] := \begin{cases}
               B[j] + (1+\eps)^d & \text{if} \; B[j] \in \big[ \eps (1+\eps)^{d-1}, \frac 1\eps (1+\eps)^d \big)\; \text{for some} \; d \in D_\ell \\
               \infty & \text{otherwise}
           \end{cases}$
    \EndFor

    \State \Return $\{ (A^{(1)},B^{(1)}),\ldots,(A^{(s)},B^{(s)})\}$
\end{algorithmic}
\end{algorithm}

Note that the condition for $B[j]$ is well-defined in the sense that it applies for at most one $d \in D_\ell$.  To see this, since two consecutive values in $D_\ell$ differ by $s$, we only need to show the inequality $\frac 1\eps (1+\eps)^d \le \eps (1+\eps)^{d + s - 1}$, which holds since $s \ge 1 + \log_{1+\eps}(1/\eps^2)$. The same can be immediately seen to hold for $A[i]$.

The size and time bounds are immediate. It remains to prove correctness.

For property (ii), consider any $i,j$ with $\frac{A[i]}{B[j]} \in [\eps,1/\eps]$. Note that there is a unique $\ell \in \{1,\ldots,s\}$ such that $A^{(\ell)}[i] \ne \infty$. For this $\ell$, we have $A^{(\ell)}[i] = (1+\eps)^d$ with $(1+\eps)^{d-1} \le A[i] < (1+\eps)^d$, for some $d \in D_\ell$. 
By the assumption $\frac{A[i]}{B[j]} \in [\eps,1/\eps]$, we obtain $\eps (1+\eps)^{d-1} \le B[j] \le \frac 1\eps (1+\eps)^d$, and thus $B^{(\ell)}[j]$ is not set to $\infty$, and we have $B^{(\ell)}[j] = B[j] + (1+\eps)^d$.
We conclude by observing that $\max\{A^{(\ell)}[i], B^{(\ell)}[j]\} = B^{(\ell)}[j] = B[j] + (1+\eps)^d \le (1+\eps)(B[j] + A[j])$. 

For property (i), consider any $i,j$ and $\ell$. 
If one of $A^{(\ell)}[i], B^{(\ell)}[j]$ is set to $\infty$, then the property holds trivially. Otherwise, we have $A^{(\ell)}[i] = (1+\eps)^d$ for some $d \in D_\ell$ and $B^{(\ell)}[j] = B[j] + (1+\eps)^{d'}$ for some $d' \in D_\ell$. We consider two cases.

\emph{Case 1: $d \le d'$.} Then $A[i] \le (1+\eps)^d \le (1+\eps)^{d'}$, and thus $B^{(\ell)}[j] = B[j] + (1+\eps)^{d'} \ge A[i] + B[j]$. 

\emph{Case 2: $d > d'$.} Then by definition of $D_\ell$ we have $d \ge d' + s$. We bound 
\begin{align} 
  B[j] \le \tfrac 1\eps (1+\eps)^{d'} \le \tfrac 1\eps (1+\eps)^{d-s} \le \tfrac 1\eps (1+\eps)^{1-s} A[i] \le \eps A[i], \label{eq:casetwo}
\end{align}
where the last inequality uses $s \ge 1 + \log_{1+\eps}(1/\eps^2)$.
This yields 
\[ A^{(\ell)}[i] \ge A[i] \stackrel{(\ref{eq:casetwo})}{\ge} (1-\eps) A[i] + B[j] \ge (1-\eps)(A[i] + B[j]). \]
In both cases we have $\max\{ A^{(\ell)}[i], B^{(\ell)}[j] \} \ge (1-\eps)(A[i] + B[j])$, which proves property (i). 
\end{proof}

\subsection{Distant Covering}
\label{distant-covering}

We now want to cover all pairs $i,j$ with $\frac{A[i]}{B[j]} \not\in [\eps,1/\eps]$. 
Our solution for this case is similar to the well-known \emph{Well-Separated Pair
Decomposition} (see~\cite{wspd,euclidean-spanners}), which we use in a one-dimensional setting and in log-scale. The
main difference is that we unite sufficiently distant pairs of the decomposition that lie on the same level.

Our constructed vectors $A^{(\ell)}$ will correspond to subsets of the entries of $A$, i.e., we have $A^{(\ell)}[i] \in \{A[i], \infty\}$, and similarly for~$B$. 
For this reason, we switch to subset notation for the majority of this section, and then return to our usual notation of vectors $A^{(\ell)},B^{(\ell)}$ in Corollary~\ref{distant-covering-cor}.

For $x,y \in \real$, we define their \emph{distance} as $d(x,y) := \max\{ \frac xy, \frac yx\}$ if $x,y < \infty$ and $d(x,\infty)=d(\infty, x) = \infty$ otherwise. For sets $X,Y \subset \real$, we define their \emph{distance} as 
$d(X,Y) := \min_{x \in X, y \in Y} d(x,y)$. 

%

\begin{lemma}[Distant Covering, Set Variant]
    \label{distant-covering-lemma}
    Given a set $Z \subset \real$ of size $n$ and a parameter $\eps > 0$, 
    there are sets $X_1,\ldots,X_s \subseteq Z$ and $Y_1,\ldots,Y_s
    \subseteq Z$ with $s = \Oh(\log{n}\log{\frac{1}{\eps}})$ such that:
    \begin{enumerate}[align=left, font=\normalfont, label=(\roman*)]
        \item for any $\ell \in [s]$ we have $d(X_\ell,Y_\ell) > \frac 1\eps$, and 
        \item for any $x,y \in Z$ with $d(x,y)\ge \frac 2\eps$ 
        and $x < y$ there is $\ell \in [s]$ such that $x \in X_\ell$ and $y\in Y_\ell$.
    \end{enumerate}
    We can compute sets $X_1,\ldots,X_s$ and $Y_1,\ldots,Y_s$ satisfying (1) and (2) time $\Oh(n\log{n}\log{\frac{1}{\eps}})$.
\end{lemma}

We will later use $Z$ as the set of all entries of vectors $A$ and~$B$.
Regarding (i), observe that if $d(x,y)>\frac 1\eps$, then the sum $x+y$ and the maximum $\max\{x,y\}$ differ by less than a factor $1+\eps$. This allows us to ensure point (i) of Sum-to-Max-Covering. Property (ii) ensures that we cover all distant pairs and thus corresponds to point (ii) of Sum-to-Max-Covering. 

\medskip
The proof outline is as follows, see also Algorithm~\ref{alg:distantcoveringlemma} for pseudocode. 
To simplify notation we assume $n$ to be a power of 2 (this is without loss of generality since we can fill up $Z$ with arbitrary numbers).
We first sort $Z$, so from now on we assume that $Z = \{z_1,\ldots,z_n\}$ with $z_1 \le \ldots \le z_n$. 
The algorithm performs $\log n$ iterations. In iteration $r$, we split $Z$ into chunks of size $n/2^r$, and we remove some chunks that are irrelevant for covering distant pairs, see procedure $\textsc{SplitChunks}$ and Figure~\ref{fig:covering-lemma-chunks}. 
Then we separate the resulting list of chunks into two sub-lists, see procedure $\textsc{SeparateChunks}$ and Figure~\ref{fig:covering-lemma-split}. 
Finally, we handle the transition between any two chunks by introducing a restricted area at their boundary, applied with $\Oh(\log \frac 1\eps)$ different shifts, see procedure $\textsc{ShiftedTransitions}$ and Figure~\ref{fig:covering-lemma-transitions}.
In the following subsections we describe the individual procedures in detail.


\begin{algorithm}
    \caption{$\textsc{DistantCoveringLemma}(Z, \eps)$}
	\label{alg:distantcoveringlemma}
\begin{algorithmic}[1]
    \State $\text{sort}(Z)$
    \Comment $Z = \{ z_1 \le z_2 \le \ldots \le z_n \}$
    \State Set $T_0$ as a list containing one element, $T_0[1] := Z$
    \For {$r = 1,2,\ldots,\ceil{\log n}$}
        \State $T_r := \textsc{SplitChunks}(T_{r-1}, \eps)$
        \State $T_{r,1}, T_{r,2} := \textsc{SeparateChunks}(T_r)$
        \State $S_{r,1} :=  \textsc{ShiftedTransitions}(T_{r,1},\eps)$
        \State $S_{r,2} :=  \textsc{ShiftedTransitions}(T_{r,2},\eps)$
    \EndFor
    \State \Return $\bigcup_r S_{r,1} \cup S_{r,2}$
\end{algorithmic}
\end{algorithm}


\subsubsection{\textsc{SplitChunks}}

Algorithm~\ref{alg:selectchunks} describes the procedure of selecting
chunks $T_r$ in every level, see also Figure~\ref{fig:covering-lemma-chunks} for an illustration. 
We start with a big chunk $T_0[1] = Z$, containing the whole input.
Then we iterate over all levels $r = 1,2,\ldots,\log n$ and construct refined chunks as follows. 
In iteration $r$, we iterate over all previous chunks $T_{r-1}[i]$. If $T_{r-1}[i]$ does not contain any two numbers in distance greater than $\frac 1\eps$, then we can ignore it. Otherwise, we split $T_{r-1}[i]$ at the middle into two chunks of half the size and add them to the list of chunks $T_r$. 
For any $r$, this yields a list of chunks $T_r$ such that 
\begin{enumerate}
  \item[(P1)] every chunk $T_r[i]$ is a subset of $Z$ of the form $\{z_a,z_{a+1},\ldots,z_b\}$ and of size $|T_r[i]| = n/2^r$, and
  \item[(P2)] every $x \in T_r[i]$ is smaller than every $y \in T_r[j]$, for any $i<j$.
\end{enumerate}
Note that at the bottom level, chunks have size 1. Moreover, for any $r > 0$ the list $T_r$ contains an even number of chunks; this will also hold for all lists of chunks constructed later.


\begin{algorithm}
    \caption{$\textsc{SplitChunks}(T_{r-1}[1 \ldots \ell], \eps)$}
	\label{alg:selectchunks}
\begin{algorithmic}[1]
    \State Initialize $T_r$ as an empty list, and $k := 1$
    \For {$i = 1,2,\ldots,\ell$} 
        \State By construction, $T_{r-1}[i] $ is of the form $\{z_a,z_{a+1},\ldots,z_b\}$ for some $a\le b$
        \If {$z_a < \eps \cdot z_b$} 
            \State $T_r[2k-1] := \{z_a,\ldots,z_{(a+b-1)/2}\}$ \\\Comment{Split $T_{r-1}[i]$ in the middle}
            \State $T_r[2k] := \{z_{(a+b+1)/2},\ldots,z_b\}$
            \State $k := k+1$
        \EndIf
    \EndFor
    \State \Return $T_r$
\end{algorithmic}
\end{algorithm}

\begin{figure}[ht!]
    \centering
   \begin{tikzpicture}[node distance=1em and 1em]
       \def\widthrec{10}
       \def\heightrec{0.2}
       \def\leveldist{1}
       \def\pertrub{0.05}

       \newcommand{\chunk}[2]{
           \pgfmathtruncatemacro\rr{2^#1}
           \draw[rounded corners=2pt, fill=blue!35] (#2 * \widthrec/\rr +
           \pertrub, -1*#1*\leveldist) rectangle (#2 * \widthrec/\rr + \widthrec/\rr - \pertrub, -1.0*#1*\leveldist - \heightrec);

           \node (lasta) at (#2 * \widthrec/\rr, -1*#1*\leveldist - \heightrec) {};
           \node (lastb) at (#2 * \widthrec/\rr + \widthrec/\rr - \pertrub, -1*#1*\leveldist - \heightrec) {};
       }

       \newcommand{\chunkwassmall}{
           \draw[decoration={brace, raise=2pt},decorate] (lastb) -- node[below=2pt] {$ < \frac{1}{\varepsilon}$ } (lasta);
       }

       \chunk{1}{0};
       \chunk{1}{1};

       \chunk{2}{0};
       \chunk{2}{1};
       \chunk{2}{2};
       \chunkwassmall;
       \chunk{2}{3};
       \chunkwassmall;

       \chunk{3}{0};
       \chunk{3}{1};
       \chunkwassmall;
       \chunk{3}{2};
       \chunk{3}{3};

       \chunk{4}{0};
       \chunk{4}{1};
       \chunk{4}{4};
       \chunk{4}{5};
       \chunk{4}{6};
       \chunk{4}{7};

       \foreach \r in  {1,2,3,4}
       {
           \node at (-1.2,-0.115-1.0*\r*\leveldist) {$r = \r$};
       }

   \end{tikzpicture}
    \caption{Illustration of the procedure \textsc{SplitChunks}, which splits and selects chunks of the input numbers on different levels~$r$.}
   \label{fig:covering-lemma-chunks}
\end{figure}
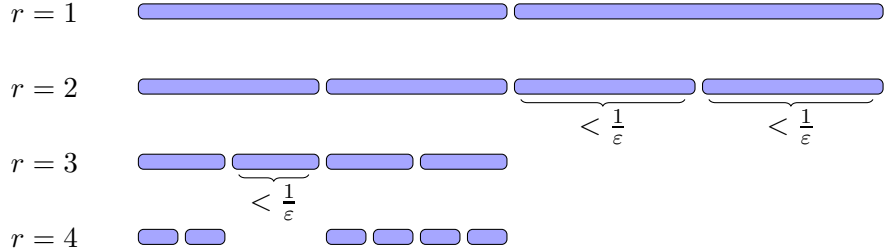


\begin{claim}
    \label{consecutive-intervals-are-distant}
    We have $d\big(T_r[2k], T_r[2k+3]\big) > \frac 1\eps$ for any level $r$ and any index $k$.
\end{claim}
\begin{proof}
    Write the parent chunk $T_r[2k+1] \cup T_r[2k+2]$ in the form $\{z_a,z_{a+1},\ldots,z_b\}$. Since $T_r[2k+1], T_r[2k+2]$ have been added, we have $z_a < \eps \cdot z_b$. Since every number in $T[2k]$ is smaller than $z_a$ and every number in $T[2k+3]$ is larger than $z_b$, the distance of $T[2k],T[2k+3]$ is greater than $\frac 1\eps$.
\end{proof}

The main property of our splitting procedure is that all $x,y \in Z$ with $d(x,y) > \frac 1\eps$ eventually are contained in consecutive chunks (see Figure~\ref{fig:covering-lemma-chunks}) -- note that we will only make use of consecutive chunks with indices $2k-1, 2k$ for some $k$ (as opposed to $2k, 2k+1$).

\begin{claim}
    \label{close-elements-in-consecutive-intervals}
    For any $x,y \in Z$, if $d(x,y) > \frac 1\eps$ and $x < y$, then
    there exist a level $r$ and index $k$ such that $x \in T_r[2k-1]$ and $y \in T_r[2k]$.
\end{claim}

\begin{proof}
    Consider the largest $r$ such that $x,y$ are contained in the same chunk $T_r[k']$. The chunk $T_r[k']$ contains at least two elements, so
    $r < \log n$. Since $d(x,y) > \frac 1\eps$, Algorithm~\ref{alg:selectchunks} splits $T_r[k']$ into chunks $T_{r+1}[2k-1]$ and $T_{r+1}[2k]$. By maximality of $r$ and by $x < y$, it follows that $x \in T_{r+1}[2k-1]$ and $y \in T_{r+1}[2k]$. This proves the claim.
\end{proof}

\subsubsection{\textsc{SeparateChunks}}

The procedure \textsc{SeparateChunks} is given a list $T_r$ of
chunks and separates it into two subsequences $T_{r,1}$ and $T_{r,2}$, where $T_{r,1}$ contains all chunks $T[i]$ with $(i \bmod 4) \in \{1,2\}$, and $T_{r,2}$ contains the remaining chunks in $T$. See Algorithm~\ref{alg:separatechunks} for pseudocode and Figure~\ref{fig:covering-lemma-split} for an illustration.

\begin{algorithm}
    \caption{$\textsc{SeparateChunks}(T_r[1 \ldots 2\ell])$}
	\label{alg:separatechunks}
\begin{algorithmic}[1]
    \State Initialize $T_{r,1},T_{r,2}$ as empty lists, and $b := 1$
    \For {$k = 1,2,\ldots,\ell$} 
        \State Append $T_r[2k-1]$ and $T_r[2k]$ to $T_{r,b}$
        \State $b := 3-b$
    \EndFor
    \State \Return $T_{r,1},T_{r,2}$
\end{algorithmic}
\end{algorithm}


\begin{figure}[ht!]
    \centering
   \begin{tikzpicture}[node distance=1em and 1em]
       \def\heightrec{0.25}
       \def\heightlevel{1}
       \def\pertrub{0.1}
       \newcounter{cA} 
       \setcounter{cA}{0}
       \newcounter{cLev} 
       \setcounter{cLev}{0}
       \node (lasta) at (0,0) {};
       \node (lastb) at (0,0) {};

       \def\redcolor{white!45!red}
       \def\bluecolor{black!25!blue}

       \newcommand{\chunk}[2]{
           \draw[rounded corners=4pt, fill=#2] (\thecA+\pertrub, \thecLev*\heightlevel) rectangle (\thecA + #1-\pertrub, \thecLev*\heightlevel + \heightrec);
           \addtocounter{cA}{#1}
       }
       \newcommand{\setleftbrace}{
           \node (lasta) at (\thecA,\heightrec + \thecLev*\heightlevel) {};
       }
       \newcommand{\setrightbrace}{
           \node (lastb) at (\thecA,\heightrec + \thecLev*\heightlevel) {};
       }

       \newcommand{\drawbrace}{
           \draw[decoration={brace,raise=2pt},decorate] (lasta) -- node[above=2pt] {$ > \frac{1}{\varepsilon}$ } (lastb);
       }

       \node at (-1,0.09) {$T_{r,1}$};
       \chunk{1}{\redcolor};
       \chunk{1}{\bluecolor};
       \setleftbrace;
       \chunk{2}{white};
       \chunk{1}{white};
       \setrightbrace;
       \drawbrace;
       \chunk{1}{\redcolor};
       \chunk{2}{\bluecolor};
       \setleftbrace;
       \chunk{1}{white};
       \chunk{1}{white};
       \setrightbrace;
       \drawbrace;

       \node at (-1,0.09-\heightlevel) {$T_{r,2}$};
       \addtocounter{cLev}{-1};
       \setcounter{cA}{0};

       \setleftbrace;
       \chunk{1}{white};
       \chunk{1}{white};
       \setrightbrace;
       \drawbrace;
       \chunk{2}{\redcolor};
       \chunk{1}{\bluecolor};
       \setleftbrace;
       \chunk{1}{white};
       \chunk{2}{white};
       \setrightbrace;
       \drawbrace;
       \chunk{1}{\redcolor};
       \chunk{1}{\bluecolor};

   \end{tikzpicture}
    \caption{Illustration of \textsc{SeparateChunks}, which separates the list of chunks $T_r$ on some level~$r$ into two sub-lists $T_{r,1}$ and $T_{r,2}$. The selected chunks are
    marked in red/light-shaded and blue/dark-shaded. (In the next step, the red/light-shaded chunks will form a set $X_\ell$, and the blue/dark-shaded ones will form a set $Y_\ell$. Note that within $T_{r,1}$ every red chunk is $\eps$-distant from every blue chunk, except for its right neighbor. This will be used by the procedure \textsc{ShiftedTransitions}.)}
   \label{fig:covering-lemma-split}
\end{figure}
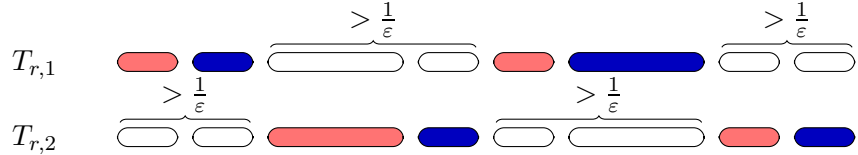


This construction ensures that consecutive chunks with indices $2k$ and $2k+1$ have distance at least $\frac 1\eps$, as shown in the following claim.

\begin{claim}
    \label{consecutive-odd-intervals-are-distant}
    We have $d\big(T_{r,b}[2k], T_{r,b}[2k+1]\big) > \receps$ for any level $r$, index $k$, and $b \in \{1,2\}$.
\end{claim}

\begin{proof}
    Because of how $T_{r,1},T_{r,2}$ are constructed, the chunks $T_{r,b}[2k]$ and $T_{r,b}[2k+1]$ correspond to chunks $T_r[2k']$ and $T_r[2k'+3]$ for some $k'$. The statement now follows from Claim~\ref{consecutive-intervals-are-distant}.
\end{proof}

The following analogue of Claim~\ref{close-elements-in-consecutive-intervals} is immediate.

\begin{claim}
    \label{close-elements-in-consecutive-intervals-after-separation}
    For any $x,y \in Z$, if $d(x,y) > \frac 1\eps$ and $x < y$, then
    there exist a level $r$, index $k$, and $b \in \{1,2\}$ such that $x \in T_{r,b}[2k-1]$ and $y \in T_{r,b}[2k]$.
\end{claim}

\begin{proof}
    Consecutive chunks $T_r[2k-1]$ and $T_r[2k]$ are either both added to $T_{r,1}$ or both added to $T_{r,2}$. The statement thus follows from Claim~\ref{close-elements-in-consecutive-intervals}.
\end{proof}

\subsubsection{\textsc{ShiftedTransitions}}

The procedure \textsc{ShiftedTransitions} is given a list of chunks $T = T_{r,b}$ and returns $\Oh(\log \frac 1\eps)$ many pairs $(X_t,Y_t)$ of the final covering
(recall the statement of Lemma~\ref{distant-covering-lemma}). Naively, we would like to assign every odd chunk to $X_t$ and every even chunk to $Y_t$, i.e., $X_t = \bigcup_{k} T[2k-1]$ and $Y_t = \bigcup_{k} T[2k]$. From
Claim~\ref{consecutive-odd-intervals-are-distant} we know that even chunks
are distant from their right neighbors, i.e.,
$d(T[2k], T[2k+1]) > \receps$.
This is not necessarily true for $d(T[2k-1],T[2k])$, and therefore we introduce a restricted area at their boundary, applied with $\Oh(\log \frac 1\eps)$ different shifts, as illustrated in Figure~\ref{fig:covering-lemma-transitions}. See Algorithm~\ref{alg:transitionintervals}.

\begin{algorithm}
    \caption{$\textsc{ShiftedTransitions}(T[1,\ldots 2\ell], \eps)$}
	\label{alg:transitionintervals}
\begin{algorithmic}[1]
    \For {$t = 0,1,\ldots,\ceil{\log_2{\receps}}$}
        \For {$k \in \{1,\ldots,\ell \}$}
            \State let $z_\textup{min}$ be the minimal number in $T[2k]$
            \State $T'[2k-1] :=  \{ z \in T[2k-1] \mid z \le \eps 2^t z_\textup{min} \}$
            \State $T'[2k] :=  \{ z \in T[2k] \mid z > 2^t z_\textup{min} \}$
        \EndFor
        \State $X_t := \bigcup_{k} T'[2k-1]$
        \State $Y_t := \bigcup_{k} T'[2k]$
    \EndFor
    \State \Return $\{(X_t,Y_t) \mid 0 \le t \le \ceil{\log_2 \receps} \}$ 
\end{algorithmic}
\end{algorithm}


\begin{figure}[ht!]
    \centering
   \begin{tikzpicture}[node distance=1em and 1em]
       \def\heightrec{1}
       \def\widthrec{5}
       \def\heightlevel{2}
       \def\pertrub{0.1}
       \def\redcolor{white!45!red}
       \def\bluecolor{black!25!blue}

       \tikzstyle{chunk style}=[rounded corners=10pt]
       \tikzstyle{lines}=[pattern=north west lines]
       \newcommand{\chunks}[3]{
           \begin{scope}[shift={(0,-#3*\heightlevel)}]
               \clip (-\widthrec+1,-1) rectangle (\widthrec -1,\heightlevel+1);
               \begin{scope}
                   \clip[chunk style] (-\widthrec,0) rectangle (-\pertrub,\heightrec);
                   \node (lasta) at (-#1, \heightrec) {};
                   \draw[fill=\redcolor] (-\widthrec,-\pertrub) rectangle (lasta);
               \end{scope}
               \begin{scope}
                   \clip[chunk style] (\pertrub,0) rectangle (\widthrec,\heightrec);
                   \node (lastb) at (#2, \heightrec) {};
                   \draw[fill=\bluecolor] (lastb) rectangle (\widthrec,0);
               \end{scope}
               \draw[chunk style] (-\widthrec,0) rectangle (-\pertrub,\heightrec);
               \draw[chunk style] (\pertrub,0) rectangle (\widthrec,\heightrec);
           \end{scope}

       }
       \newcommand{\setleftbrace}{
           \node (lasta) at (\thecA,\heightrec + \thecLev*\heightlevel) {};
       }
       \newcommand{\setrightbrace}{
           \node (lastb) at (\thecA,\heightrec + \thecLev*\heightlevel) {};
       }

       \newcommand{\drawbrace}[1]{
           \draw[decoration={brace,raise=2pt},decorate] (lasta) -- node[above=2pt] {#1} (lastb);
       }

       \node at (-\widthrec/2,2) {$T[2k-1]$};
       \node at (\widthrec/2,2) {$T[2k]$};

       \chunks{2}{0}{0};
       \node (emint) at (lasta) {};
       \drawbrace{$1/\varepsilon$};

       \chunks{1.33}{0.66}{1};
       \drawbrace{$1/\varepsilon$};

       \begin{scope}
           \clip[chunk style] (\pertrub,-\heightlevel) rectangle (\widthrec,\heightrec-\heightlevel);
           \fill[lines] ($(lastb)$) rectangle ($(lastb) + (-0.66,-\heightrec)$) {};
       \end{scope}
       \node (lastb) at (lasta) {};
       \node (lasta) at ($(emint) - (0,\heightlevel)$) {};
       \fill[lines] ($(lasta)$) rectangle ($(lasta) + (0.66,-\heightrec)$) {};

       \drawbrace{$2$};

       \node at (0,-1.25* \heightlevel) {$\vdots$};

       \chunks{0}{2}{2.25};
       \drawbrace{$1/\varepsilon$};
       \begin{scope}
           \clip[chunk style] (-\widthrec,-2.25*\heightlevel) rectangle (-\pertrub,\heightrec-2.25*\heightlevel);
           \fill[lines] (-0.66,0) rectangle (0,-3.25*\heightlevel) {};
       \end{scope}
       \fill[lines] (2-0.66,-2.25*\heightlevel) rectangle (2,-2.25*\heightlevel+\heightrec) {};

   \end{tikzpicture}
    \caption{Illustration of \textsc{ShiftedTransitions} in log-scale. Dashed areas
    represent the added/removed numbers from iteration $t$ to $t+1$. In each
    iteration, we shift to the right by a factor 2, resulting in at most 
    $\ceil{\log_2{\frac{1}{\varepsilon}}}$ iterations. Note that in each iteration the red/light-shaded numbers are in distance greater than $\frac 1\eps$ from the blue/dark-shaded numbers. Moreover, 
    the distance between any two numbers in the dashed area is less than
    $\frac{2}{\varepsilon}$.}
   \label{fig:covering-lemma-transitions}
\end{figure}
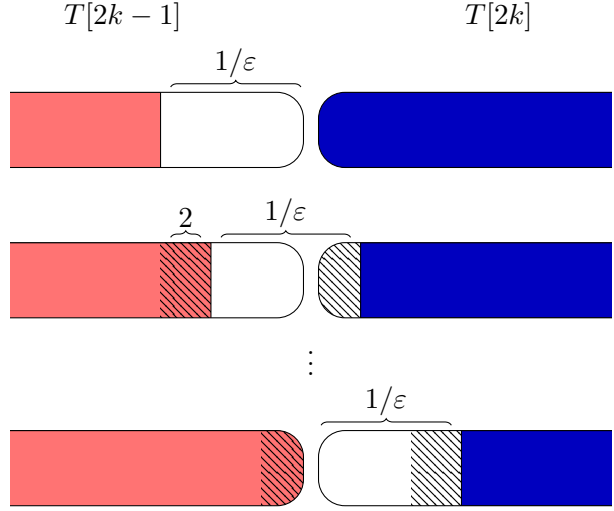


\begin{claim}
    \label{covering-lemma-correctness}
    For any output $(X_t,Y_t)$ by $\textsc{ShiftedTransitions}(T_{r,b},\eps)$ we have $d(X_t,Y_t) > \frac 1\eps$.
\end{claim}

\begin{proof}
    Let $T = T_{r,b}$. By Claim~\ref{consecutive-odd-intervals-are-distant} and
    sortedness (see property (P2)), any chunks $T[i]$ and $T[j]$ with $j \ge i+2$ have distance greater than $\frac 1\eps$. 
    Since $X_t$ only contains numbers from odd chunks $T[2k-1]$, and $Y_t$ only contains numbers from even chunks $T[2k]$, we obtain that any $x \in X_t, y \in Y_t$ within distance $\frac 1\eps$ satisfy $x \in T[2k-1], y \in T[2k]$ for some index $k$. However, in any iteration $t$ the subsets $T'[2k-1] \subseteq T[2k-1]$ and $T'[2k] \subseteq T[2k]$ are chosen to have distance greater than $\frac 1\eps$, and hence $d(X_t,Y_t) > \frac 1\eps$.
%
%
\end{proof}

\begin{claim}
    \label{covering-lemma-completness}
    For any $x,y \in Z$, if $d(x,y) \ge \frac 2\eps$ and $x < y$, then
    there exist a level $r$ and $b \in \{1,2\}$ such that $\textsc{ShiftedTransitions}(T_{r,b},\eps)$ outputs a pair $(X_t,Y_t)$ with $x \in X_t, y \in Y_t$.
%
%
\end{claim}

\begin{proof}
    By Claim~\ref{close-elements-in-consecutive-intervals-after-separation}, there are $r,k,b$ with $x \in T_{r,b}[2k-1]$ and $y \in T_{r,b}[2k]$. 
    Let $T = T_{r,b}$ and let $z_\textup{min}$ be the minimal number in $T[2k]$. 
    If $x \le \eps \cdot z_\textup{min}$, then in iteration $t=0$ we construct $T'[2k-1]$ containing $x$, and $T'[2k] = T[2k]$ contains $y$, so $x \in X_0$ and $y \in Y_0$. 
    Otherwise, let $t \in \mathbb{N}$ be minimal with $x \le \eps 2^t \cdot
z_\textup{min}$. By sortedness (see property (P2)) we have $x < z_\textup{min}$
and thus $t \le \ceil{\log_2 \receps}$. 
    Hence, in iteration $t$ the set $T'[2k-1]$ contains $x$. Moreover, by minimality of $t$ and $d(x,y) \ge \frac 2\eps$ we have $\eps 2^{t-1} \cdot z_\textup{min} < x \le \frac \eps 2 y$, and thus $y > 2^t \cdot z_\textup{min}$, so $y$ is contained in $T'[2k]$, yielding $x \in X_t, y \in Y_t$.
%
%
\end{proof}

\subsubsection{Proof of Distant Covering}

\begin{proof}[Proof of Lemma~\ref{distant-covering-lemma}]
    Properties (i) and (ii) follow immediately from Claims~\ref{covering-lemma-correctness} and \ref{covering-lemma-completness}.
    The bound $s = \Oh(\log{n}\log{\frac{1}{\eps}})$ on the number of constructed pairs $(X_\ell,Y_\ell)$ is immediate from the loops in Algorithms~\ref{alg:distantcoveringlemma} and \ref{alg:transitionintervals}.
    Finally, the running time of $\Oh(n\log{n}\log{\frac{1}{\eps}})$ is immediate from inspecting the pseudocode.
%
%
%
\end{proof}

It remains to translate Lemma~\ref{distant-covering-lemma} to the vector notation of Sum-to-Max-Covering.

\begin{corollary}[Distant Covering, Vector Variant]
    \label{distant-covering-cor}
    Given vectors $A,B \in \real^n$ and a parameter $\eps >
    0$, there are vectors 
    $A^{(1)},\ldots,A^{(s)}, B^{(1)},\ldots,B^{(s)} \in \real^n$ with $s = \Oh(\log n \log \frac{1}{\eps})$ such that:
    \begin{enumerate}[align=left, font=\normalfont, label=(\roman*)]
        \item for all $i,j \in [n]$ and all $\ell \in [s]\colon \; \max\{A^{(\ell)}[i],B^{(\ell)}[j]\} \ge (1-2\eps)(A[i] + B[j])$, and
        \item for all $i,j \in [n]$ if $\frac{A[i]}{B[j]} \not\in [\eps,1/\eps]$ then $\exists \ell \in [s]\colon \; \max\{A^{(\ell)}[i],B^{(\ell)}[j]\} \le A[i] + B[j]$.
    \end{enumerate}
    We can compute such vectors $A^{(1)},\ldots,A^{(s)}, B^{(1)},\ldots,B^{(s)}$ in time $\Oh(n \log n \log \frac{1}{\eps})$.
%
%
%
\end{corollary}

\begin{proof}
    Set $Z := \{A[i], B[i] \mid i \in [n]\}$ and run Lemma~\ref{distant-covering-lemma} on $(Z,\eps')$ with $\eps' := 2\eps$ to obtain subsets $X_1,\ldots,X_s \subseteq Z$ and $Y_1,\ldots,Y_s \subseteq Z$ with $s = \Oh(\log{n}\log{\frac{1}{\eps}})$. 
    We only double the number of subsets by considering $(X'_1,\ldots,X'_{2s}) := (X_1,\ldots,X_s,Y_1,\ldots,Y_s)$ and $(Y'_1,\ldots,Y'_{2s}) := (Y_1,\ldots,Y_s,X_1,\ldots,X_s)$. 
    We construct vectors $A^{(\ell)},B^{(\ell)}$ with $\ell \in [2s]$ by setting
    \begin{align*}
      A^{(\ell)}[i] := \begin{cases} A[i] & \text{if} \; A[i] \in X'_\ell, \\ \infty & \text{otherwise}, \end{cases} 
      \qquad
      B^{(\ell)}[i] := \begin{cases} B[i] & \text{if} \; B[i] \in Y'_\ell, \\ \infty & \text{otherwise.} \end{cases}
    \end{align*}
    The size and time bounds are immediate. 
    
    For any numbers $x,y\in \real$ with $d(x,y) > \frac 1\eps$ we claim that $\max\{x,y\} \ge (1-\eps)(x+y)$. Indeed, assume without loss of generality $x < \eps y$, then we have $\max\{x,y\} = y > (1-\eps)y + x \ge (1-\eps)(x+y)$. 
    
    Property (i) of Lemma~\ref{distant-covering-lemma} yields that $d(A^{(\ell)}[i], B^{(\ell)}[j]) > \frac 1{\eps'}$ for any $i,j,\ell$. Hence, we have
    \begin{displaymath} \max\{ A^{(\ell)}[i], B^{(\ell)}[j] \} \ge (1-\eps')(A^{(\ell)}[i] + B^{(\ell)}[j])
    \ge (1-\eps')(A[i]+B[j]) = (1-2\eps)(A[i]+B[j]), \end{displaymath}
    proving property (i) of the corollary.
    
    Note that by switching from $X_1,\ldots,X_s$ and $Y_1,\ldots,Y_s$ to $X'_1,\ldots,X'_{2s}$ and $Y'_1,\ldots,Y'_{2s}$
    we made Lemma~\ref{distant-covering-lemma} symmetric, and thus the requirement $x<y$ of its property (ii) is removed. 
    Thus, property (ii) of  Lemma~\ref{distant-covering-lemma} yields that for any $i,j$ with $d(A[i], B[j]) \ge \frac 2{\eps'} = \frac 1 \eps$ there exists $\ell \in [2s]$ with $A^{(\ell)}[i] = A[i]$ and $B^{(\ell)}[j] = B[j]$, and thus 
    \[ \max\{A^{(\ell)}[i], B^{(\ell)}[j]\} = \max\{A[i],B[j]\} \le A[i] + B[j], \] 
    proving property (ii) of the corollary.
\end{proof}

\subsection{Proof of Sum-To-Max-Covering}
\label{combining-covering}


The following variant of Sum-To-Max-Covering follows immediately from combining Close Covering
(Lemma~\ref{close-covering}) and Distant Covering
(Lemma~\ref{distant-covering-lemma}). Note that compared to Lemma~\ref{covering-lemma} there is an additional factor $1-2\eps$ on the right hand side of property (i).

\begin{lemma}[Weaker Sum-to-Max-Covering]
    \label{covering-lemma-variant}
    Given vectors $A,B \in \real^n$ and a parameter $\eps >
    0$, there are vectors $A^{(1)},\ldots,A^{(s)}, B^{(1)},\ldots,B^{(s)} \in \real^n$
    with $s = \Oh(\frac{1}{\eps}\log{\frac{1}{\eps}} + \log{n}\log{\frac{1}{\eps}})$ such that:
    \begin{enumerate}[align=left, font=\normalfont, label=(\roman*)]
    	 \item for all $i,j \in [n]$ and all $\ell \in [s]$: \[\max\{A^{(\ell)}[i],B^{(\ell)}[j]\} \ge (1-2\eps)(A[i] + B[j]),\]
    	 \item for all $i,j \in [n]$ there exists $\ell \in [s]$: \[\max\{A^{(\ell)}[i],B^{(\ell)}[j]\} \le (1+\eps)(A[i] + B[j]).\]
    \end{enumerate}
    We can compute such vectors $A^{(1)},\ldots,A^{(s)}, B^{(1)},\ldots,B^{(s)}$ in time $\Oh(\frac{n}{\eps}\log{\frac{1}{\eps}} +
    n\log{n}\log{\frac{1}{\eps}})$.
\end{lemma}

\begin{proof}[Proof of Lemma~\ref{covering-lemma-variant}]
    Given vectors $A,B \in \real^n$ and a parameter $\eps > 0$, we simply run Close Covering and Distant Covering on $(A,B,\eps)$ and concatenate the results, see Algorithm~\ref{alg:covering-complete-weak}. Correctness as well as size and time bounds are immediate consequences of Lemma~\ref{close-covering} and Corollary~\ref{distant-covering-cor}. 
%
%
%
\begin{algorithm}
    \caption{$\textsc{WeakerSumToMaxCovering}(A,B,\eps)$.}
	\label{alg:covering-complete-weak}
\begin{algorithmic}[1]
    \State \Return $\textsc{DistantCovering}(A,B,\eps) \cup \textsc{CloseCovering}(A,B,\eps)$
\end{algorithmic}
\end{algorithm}
%
\end{proof}

\begin{proof}[Proof of Lemma~\ref{covering-lemma}]
  To prove the stronger variant given in the beginning of Section~\ref{sec:covering-lemma}, we have to remove the factor $1-2\eps$ on the right hand side of property (i) in Lemma~\ref{covering-lemma-variant}. To this end, given $A,B,\eps$, we run the construction from Lemma~\ref{covering-lemma-variant} on $A,B,\eps'$ with $\eps' := \frac \eps 5$, and we divide every entry of the resulting vectors by $1-2\eps'$, see Algorithm~\ref{alg:covering-complete}.
  For correctness, note that the division by $1-2\eps'$ removes the factor $1-2\eps'$ from property (i) in Lemma~\ref{covering-lemma-variant}, i.e., we obtain the claimed $\max\{A^{(\ell)}[i],B^{(\ell)}[j]\} \ge A[i] + B[j]$ for all $i,j,\ell$. For property (ii), note that the division by $1-2\eps'$ leaves us with $\max\{A^{(\ell)}[i],B^{(\ell)}[j]\} \le \frac{1+\eps'}{1-2\eps'}(A[i] + B[j])$ for all $i,j$ and some $\ell$. Using $\eps' = \frac \eps 5$ and $\frac{1+\eps/5}{1-2 \eps/5} \le 1+\eps$ for any $\eps \in (0,1]$ finishes the proof.
  \begin{algorithm}
    \caption{$\textsc{SumToMaxCovering}(A,B,\eps)$.}
	\label{alg:covering-complete}
\begin{algorithmic}[1]
    \State \Return $\frac 1{1-2\eps/5} \cdot \textsc{WeakerSumToMaxCovering}(A,B,\frac \eps 5)$
\end{algorithmic}
\end{algorithm}
\end{proof}

\section{Strongly Polynomial Approximation for Undirected APSP}
\label{strong-apsp}

In this section, we present a strongly polynomial $(1+\eps)$-approximation for APSP on undirected graphs that runs in time $\tOh(n^\omega / \eps)$. 
Contrary to the title of this paper our algorithm will use the scaling technique, but we combine it with edge contractions and amortized analysis to avoid the factor $\log W$; this is inspired by \citet{eva-tardos-strongly} and uses similar edge contraction arguments as~\cite{KleinS92,Cohen97}.
We start by describing the previously fastest approximation algorithm for APSP (Section~\ref{sec:zwick-algorithm}), and then present our adaptations (Section~\ref{undirected-apx-apsp}).

\subsection{Zwick's Approximation for APSP}
\label{sec:zwick-algorithm}

\citet{zwick-apsp} obtained an $\Ot(\frac{n^\omega}{\eps} \log{W})$-time $(1+\eps)$-approximation algorithm for (directed or undirected) APSP as follows. The first step is to reduce approximate APSP to approximate \minprod, see Theorem~\ref{thm:red-apsp-minprod}. It is well-known that \minprod on $n \times n$-matrices with entries in $\{1,\ldots,W\}$ can be solved exactly in time $\tOh(W n^\omega)$. Zwick utilized this fact via \emph{adaptive scaling} to realize his second step, as shown in Algorithm~\ref{alg:approx-minprod}.

\begin{algorithm}
    \caption{$\textsc{scale}(A,q,\eps)$.}
	\label{alg:scaling}
\begin{algorithmic}[1]
	\State $A'[i,j] = \begin{cases} \ceil{A[i,j]/(\eps 2^q)} & \text{if} \; A[i,j] \le 2^q \\
		\infty & \text{otherwise} \\
	\end{cases}$
	\State \Return $A'$
\end{algorithmic}
\end{algorithm}

\begin{algorithm}
    \caption{$\textsc{Zwick-Apx-MinProd}(A,B,\eps)$.}
	\label{alg:approx-minprod}
\begin{algorithmic}[1]
    \State Initialize $\tilde C[i,j] := \infty$ for all $i,j$
    \State Let $W$ be the largest entry of $A,B$
	\For {$q=0,1,2,\ldots, \ceil{\log{W}}+1$}
	  \State $A' = \textsc{scale}(A,q,\eps)$ \Comment{Scale matrix $A$ so entries are from $\{0,\ldots,\ceil{\receps}\}$}
	  \State $B' = \textsc{scale}(B,q,\eps)$
	  \State $C' = \textsc{MinPlusProd}(A',B')$ \Comment{This works in time $\Ot(\frac{n^\omega}{\eps})$}
	  \State $\tilde C[i,j] = \min \{C[i,j], \eps 2^q C'[i,j]\}$ for all $i,j$
    \EndFor
    \State \Return $\tilde C$
\end{algorithmic}
\end{algorithm}

In each iteration $q$ of Algorithm~\ref{alg:approx-minprod} the entries that are
greater than $2^q$ are ignored (they are replaced by $\infty$). The remaining
entries are scaled and rounded to lie in the range $\{1,\ldots,\ceil{\receps}\}$ by
Algorithm~\ref{alg:scaling}. This yields scaled matrices $A',B'$ with integer entries bounded by $\Oh(\receps)$, so their \minprod $C'$ can be computed in time $\tOh(\frac{n^\omega}{\eps})$. The output $\tilde  C$ of the algorithm is the entry-wise minimum of all computed products $C'$ over all $q$. 
The total running time is $\tOh(\frac{n^\omega}{\eps} \log W)$.

Since we always round up, one can see that each entry of $\tilde C$ is at least
the corresponding entry of the correct \minprod $C$ of $A,B$. For the other
direction, for all $i,j$ there is an iteration $q$ such that $2^{q-1} < C[i,j]
\le 2^q$. Consider any $k$ with $C[i,j] = A[i,k] + B[k,j]$. Then $A[i,k],B[k,j]
\le 2^q$, so they are not set to $\infty$. We obtain that $C'[i,j] \le A'[i,k] +
B'[k,j] \le \frac{A[i,k] + B[k,j]}{\eps 2^q} + 2$, and thus $\tilde C[i,j] \le
\eps 2^q C'[i,j] \le C[i,j] + \eps 2^{q+1} \le (1+4\eps) C[i,j]$. Replacing $\eps$ by $\eps/4$ yields the claimed approximation.


\subsection{Undirected APSP in Strongly Polynomial Matrix-Multiplication Time}
\label{undirected-apx-apsp}

We now essentially remove the $\log W$-factor from Zwick's algorithm for undirected graphs, proving the following restated theorem from the introduction.

\begingroup
\def\thetheorem{\ref{thm:undir-apsp}}
\begin{theorem}
    $(1+\eps)$-Approximate Undirected APSP is in strongly polynomial time\footnote{Here, by time we mean the number of \emph{arithmetic operations} performed on a RAM machine. The \emph{bit complexity} of the algorithm is bounded by the number of arithmetic operations times $\log \log W$ (up to terms hidden by $\tOh$).} $\Ot(\frac{n^\omega}{\eps})$.
\end{theorem}
\addtocounter{theorem}{-1}
\endgroup

\begin{proof}[Proof of Theorem~\ref{thm:undir-apsp}]
The algorithm proceeds in iterations $q = 1, 2, 4,\ldots$ up to the largest power of 2 bounded by $nW$. In iteration $q$, the goal is to find all shortest paths of length in $[q,2q)$. For this, all edges of $G$ of weight at least $2q$ are irrelevant. Therefore, in the first iteration we start with an empty graph $H$, and in iteration $q$ we add all edges of $G$ with weight in $[q,2q)$ to $H$. We then round down all edge weights to multiples of $q \eps/n$. This may result in edges of weight 0, which we contract (this crucially uses that we consider undirected graphs). Finally, we run Zwick's algorithm on $H$ and update the corresponding distances. Specifically, if we compute a distance $\tilde d_H(u,v) \in [(1-\eps)q,(1+\eps)2q)$ in~$H$, then we iterate over all vertices $i$ in $G$ that were contracted to $u$ in $H$, and similarly over all $j$ that were contracted to $v$, and we update our estimated distance $D[i,j]$. See Algorithm~\ref{alg:approx-uAPSP}.

\begin{algorithm}
	\caption{$\textsc{ApproximateUndirectedAPSP}(G,\eps)$.}
	\label{alg:approx-uAPSP}
\begin{algorithmic}[1]
	\State Initialize $H$ to be the graph with $n$ isolated nodes, i.e., $H = (V(G),\emptyset)$
	\State Initialize $D[i,j] := \infty$ for all $i,j \in V(G)$
	\For {$q=1,2,4,\ldots, 2^{\lfloor \log(nW) \rfloor}$}
	  \Comment{Find all shortest paths of length in $[q,2q)$:}
	  \State Add all edges of $G$ with weight in $[q,2q)$ to $H$
	  \State Round down all edge weights of $H$ to multiples of $\frac{q \eps}{n}$
	  \State Contract all edges of $H$ with weight 0
	  \State Run Zwick's $(1+\eps)$-approximation for APSP on $H$, obtaining distances $\tilde d_H(u,v)$
	  \For{all nodes $u,v$ of $H$ with $\tilde d_H(u,v) \in [(1-\eps)q,(1+\eps)2q)$} 
	    \State $D[i,j] := \min\{D[i,j], \tilde d_H(u,v)\}$ for every node $i$ ($j$) of $G$ that was contracted to $u$ ($v$), 
      \EndFor
    \EndFor
    \State Return $D$
\end{algorithmic}
\end{algorithm}

\paragraph{Correctness}
Denote by $d_G(i,j)$ the correct distance between $i$ and $j$ in $G$, and by $d_H(.,.)$ the correct distance in $H$. The computed approximation satisfies $d_H(u,v) \le \tilde d_H(u,v) \le (1+\eps) d_H(u,v)$ for any $u,v \in V(H)$. 
%
%
\begin{claim}
	\label{correctness-apsp-claim}
  Consider $i,j \in V(G)$ that have been contracted to $u,v$ in $H$, respectively. If we have $\tilde d_H(u,v) \ge (1-\eps)q$ then $\tilde d_H(u,v) \ge (1-\eps) d_G(i,j)$.
\end{claim}

Since we only update distances when $\tilde d_H(u,v) \in [(1-\eps)q,(1+\eps)2q)$, by this claim the output of Algorithm~\ref{alg:approx-uAPSP} satisfies $D[i,j] \ge (1-\eps) d_G(i,j)$ for all $i,j \in V(G)$. 

\begin{proof}[Proof of Claim~\ref{correctness-apsp-claim}]
  Consider a path $P$ from $u$ to $v$ in $H$ realizing $d_H(u,v)$. Uncontract all contracted edges of $H$ to obtain a graph $H'$. Since we only contracted edges of (rounded) weight 0, the path $P$ corresponds to a path $P'$ from $i$ to $j$ in $H'$ of (rounded) length $d_H(u,v)$. Since we can assume $P'$ to be a simple path and we rounded down edge weights to multiples of $q \eps/n$, in the graph $G$ path $P'$ has length at most $d_H(u,v) + q\eps$. Hence, we have $d_G(i,j) \le d_H(u,v) + q\eps$.

  Note that the claim is trivial if $d_G(i,j) \le q$, so assume $d_G(i,j) > q$. Then we conclude by bounding
  $\tilde d_H(u,v) \ge d_H(u,v) \ge d_G(i,j) - q\eps > ( 1 - \eps ) d_G(i,j)$. 
\end{proof}

It remains to show $D[i,j] \le (1+\eps) d_G(i,j)$ for all $i,j \in V(G)$.
Consider the iteration $q$ with $d_G(i,j) \in [q,2q)$. 
Note that all edges of any shortest path between $i$ and $j$ are added in or before iteration $q$.
Let $u,v$ be the nodes that $i,j$ have been contracted to in $H$ in iteration $q$, respectively. 
Since rounding down to multiples of $q\eps/n$ reduces the distance by at most $q\eps$, and since contracting edges of rounded weight 0 does not change any distances, we have $d_H(u,v) \in [(1-\eps)d_G(i,j), d_G(i,j)]$. This yields $d_H(u,v) \in [(1-\eps)q,2q)$, and in particular we have $u \ne v$. By the properties of the approximation (i.e., $d_H(u,v) \le \tilde d_H(u,v) \le (1+\eps) d_H(u,v)$), we obtain $\tilde d_H(u,v) \in [(1-\eps)q,(1+\eps)2q)$. This triggers the update of $D[i,j]$, which yields $D[i,j] \le \tilde d_H(u,v) \le (1+\eps) d_H(u,v) \le (1+\eps) d_G(i,j)$. 

In total, we obtain that $(1-\eps) d_G(i,j) \le D[i,j] \le (1+\eps) d_G(i,j)$ for all $i,j \in V(G)$. By dividing all output entries $D[i,j]$ by $(1-\eps)$, we may instead obtain a ``one-sided'' approximation of the form $d_G(i,j) \le D'[i,j] \le (1+\Oh(\eps)) d_G(i,j)$.

\paragraph{Running Time}
We call an iteration $q$ \emph{void} if $G$ has no edge with weight in $[q,2q)$ and $H$ is empty (i.e., $H$ consists of isolated vertices). 
Observe that void iterations do not change $H$ or $D$.
In order to avoid the $\log W$-factor of a naive implementation of Algorithm~\ref{alg:approx-uAPSP}, we skip over all void iterations.

In iteration $q$, let $C_{q,1},\ldots,C_{q,k(q)}$ be all connected components of $H$ of size at least 2, and let $n_{q,1},\ldots,n_{q,k(q)}$ be their sizes (i.e., number of vertices). Note that instead of running Zwick's algorithm on $H$, it suffices to run it on each $C_{q,i}$. Moreover, after rounding, the ratio of the largest to smallest weight in $H$ is $\Oh(n/\eps)$, and thus the $\log W$ factor of Zwick's algorithm is $\Oh(\log(n/\eps))$. Hence, we can bound the contribution of Zwick's algorithm to our running time by $\tOh\big(\sum_{q,i} (n_{q,i}^\omega / \eps) \log (n/\eps)\big)$. 


The life-cycle of every pair of vertices $u,v
\in V(G)$ can be described by the following states: (1) $u$ and $v$
are not connected in $H$; (2) $u$ and $v$ are connected in $H$; (3) $u$ and $v$ are
contracted into the same vertex of $H$. Observe that each pair $u,v \in V(G)$ can be in state (2) for at most
$\Oh(\log(n/\eps))$ iterations. Indeed, if $u$ and $v$ are connected in
iteration $q$, then in iteration $q' > \frac{n}{\eps} q$ they have been
contracted. It follows that $\sum_{q,i} n_{q,i}^2 = \Oh(n^2 \log(n/\eps))$, since the former counts the pairs of connected nodes in $H$, summed over all iterations $q$.

Combining this with our running time bound of $\tOh\big(\sum_{q,i} (n_{q,i}^\omega / \eps) \log (n/\eps)\big)$, the fact $\omega \ge 2$, and Jensen's inequality, yields a time bound of $\tOh(n^\omega / \eps)$. This bounds the running time spent in calls to Zwick's algorithm. Most other parts of our algorithm can be seen to take time $\tOh(n^2)$ in total. 

A subtle point is the update of matrix entries $D[i,j]$ in line 9 of Algorithm~\ref{alg:approx-uAPSP}. 
Consider any $i,j$ and let $u,v$ be the vertices that $i,j$ have been contracted to in $H$ in iteration $q$, respectively. 
Note that $D[i,j]$ can only change if $u \ne v$ and $u$ and $v$ are connected in $H$. This situation can only happen for $\Oh(\log(n/\eps))$ iterations, since then $u$ and $v$ will be contracted to the same vertex. Hence, in total updating $D$ takes time $\Oh(n^2 \log(n/\eps))$. This finishes the proof.
\end{proof}

\section{Strongly Polynomial Approximation for Graph Characteristics}
\label{applications}

One of the fundamental challenges in network science is the identification of ``important'' or ``central'' nodes
in a network. Different \emph{graph
characteristics} have been proposed to capture this notion~\cite{freeman1978centrality}.  For example the \emph{Median} of a graph
is a node that minimizes the sum of the distances to all other nodes in graph,
the \emph{Center} of a graph is a node that minimizes the maximum distance to
any other node (this distance is called Radius) and the \emph{Diameter} of a
graph is the distance of the furthest pair of vertices in the graph. Centrality
measures are actively generalized to weighted
graphs~\cite{weighted-centrality}. In this section, we present a 
simple argument that yields strongly polynomial approximation schemes for these problems. The following theorem is restated from the introduction.

\Applications*

\citet{abboud-centrality} observed that \emph{Diameter, Radius} and \emph{Median} admit
$\Ot(\frac{n^\omega}{\eps}\log{W})$-time approximation schemes via Zwick's approximation of APSP. Similarly, \citet{minweightcycles} observed that
\emph{Minimum Weight Triangle} admits an
$\Ot(\frac{n^\omega}{\eps}\log{W})$-time approximation scheme. They used this as a black-box to show that
\emph{Minimum Weight Cycle} (both in directed and undirected graphs) admits
an $\Ot(\frac{n^\omega}{\eps}\log{W})$-time approximation scheme. 

\begin{proof}[Proof of Theorem~\ref{thm:applications}]
    Let $G$ be a given graph. For any number $w \in \real$, define $G_{w}$ as
    the graph $G$ where we remove all edges of weight $>w$ and change the weight
    of all remaining edges to $w$. On $G_{w}$ we can compute a 2-approximation
    for each of the considered problems in time $\Ot(n^\omega)$ (since $\eps =
    1$ is constant and there are only two different edge weights, so also W is
    constant). Note that if the result on $G_{w}$ is infinite, then the solution
    value on $G$ is greater than $w$, as we need to include at least one edge of
    weight greater than $w$. Moreover, if the solution value on $G_{w}$ is
    finite, then it is at most $w \cdot n^2$, since this is the total weight of
    all edges in $G_{w}$. In particular, this means that the solution value of $G$ is at most
    $w n^2$.
    
    We use this as follows. First, we sort all edge weights of $G$ and perform
    binary search to determine the smallest edge weight $w$ of $G$ such that the
    solution value on $G_{w}$ is finite. It follows that the solution value on
    $G$ is in $[w,w n^2]$, so we have an $\Oh(n^2)$-approximation.  Now we round
    up all edge weights of $G$ to multiples of $w\eps/n^2$. This changes the
    total edge weight of $G$ by at most $\eps w$, and thus also the weight of an
    optimal solution by at most a multiplicative factor $1+\eps$.  The ratio
    between the largest and smallest weight in the resulting graph $G'$ is at
    most $W \le \frac{n^4}{\eps}$. Hence the 
    $\Ot(\frac{n^\omega}{\eps} \log{W})$-time approximation scheme runs in time $\Ot(\frac{n^\omega}{\eps})$ on $G'$. 
\end{proof}

\section{Strongly Polynomial Approximation for \minconv}
\label{minconv}

In this section, we consider sequences $A \in \real^n$ and index their entries by $A[0],\ldots,A[n-1]$. 
Given two sequences $A,B\in \real^n$, the \emph{convolution problem in the
$(\otimes,\oplus)$-semiring} is to compute the sequence $C \in \real^n$ with $C[k] = \bigoplus_{i+j=k} (A[i] \otimes
B[j])$
for all $0 \le k < n$. 
Clearly, the problem can be solved using $\Oh(n^2)$ ring operations.
In the standard $(\cdot,+)$-ring, the problem can be solved in time $\Oh(n\log{n})$ by Fast Fourier Transform (FFT). 

\emph{\minmaxconv} is the problem of computing convolution in the (min,max)-semiring.
\citet{first-minmaxconv} was the first to design a subquadratic-time algorithm for this problem, obtaining time 
$\Oh(n^{3/2}\sqrt{\log{n}})$. He conjectured that his algorithm can be improved to time
$\Ot(n)$ but so far no improvement has been found. 


\emph{\minconv} is the problem of computing convolution in the (min,+)-semiring.
Computing \minconv in time $\Oh(n^{2-\delta})$ for any $\delta>0$ is a major open
problem~\cite{cygan-icalp2017,kunnemann-icalp2017}. \citet{tree-sparsity} were
the first to study $(1+\eps)$-approximation algorithms for \minconv. They obtained an $\Oh(\frac{n}{\eps^2}\log{n}\log^2{W})$-time algorithm, which they used to
design an approximation algorithm for \sparsity. Subsequently, their running time was improved to $\Oh(\frac{n}{\eps}\log(n/\eps)\log{W})$~\cite{partition}.


\medskip
In this section, we start with a simple $(1+\eps)$-approximation algorithm for \minconv that directly follows from our Sum-To-Max-Covering and runs in time $\tOh(n^{3/2} / \eps)$, see Theorem~\ref{thm:minconvapxsimple}. This is the first strongly polynomial $(1+\eps)$-approximation algorithm for this problem (with a running time of $\Oh(n^{2-\delta})$ for any $\delta > 0$).
We then prove an equivalence of approximate \minconv and exact \minmaxconv, see Theorem~\ref{thm:conv-equ}.
Finally, we use more problem-specific arguments to obtain an improved approximation algorithm running in time $\tOh(n^{3/2} / \eps^{1/2})$, see Theorem~\ref{apx-minconv}.
%

\subsection{Simple Approximation Algorithm}

Direct application of our Sum-to-Max-Covering yields the following strongly polynomial approximation algorithm for \minconv, similarly as for APSP.

\begin{theorem}
    \label{thm:minconvapxsimple}
    $(1+\eps)$-Approximate \minconv can be solved in strongly polynomial time
    $\Ot(n^{3/2}/\eps)$.
\end{theorem}

\begin{proof}
  Consider input sequences $(A[i])_{i=0}^{n-1},\, (B[i])_{i=0}^{n-1}$ on which we want to compute $(C[k])_{k=0}^{n-1}$ with $C[k] = \min_{i+j=k} (A[i] + B[j])$. 
  We view the sequences $A,B$ as vectors in $\real^n$, in order to apply Sum-To-Max-Covering (Lemma~\ref{covering-lemma}). This yields sequences $A^{(1)},\ldots,A^{(s)},\, B^{(1)},\ldots,B^{(s)} \in \real^n$ with $s = \Oh(\eps^{-1} \polylog(n/\eps))$. We compute the \minmaxconv of every layer $A^{(\ell)},B^{(\ell)}$. (Note that we can first replace the entries of $A^{(\ell)},B^{(\ell)}$ by their ranks; this is necessary since the input format for approximate \minconv is floating-point, but \minmaxconv requires standard bit representation of integers.) We then return the entry-wise minimum of the results, see Algorithm~\ref{alg:approx-minconv-simple}. 
  The output sequence $\tilde C$ is a $(1+\eps)$-approximation of $C$; this follows from the properties of Sum-To-Max-Covering, analogously to the proof of Theorem~\ref{thm:apx-minprod}.
%
%
  Since Sum-To-Max-Covering takes time $\tOh(n/\eps)$, the running time is dominated by computing $s$ times a \minmaxconv, resulting in $\tOh(s n^{3/2}) = \tOh(n^{3/2} / \eps)$.
  \begin{algorithm}
    \caption{$\textsc{ApproximateMinConv}(A,B,\eps)$}
	\label{alg:approx-minconv-simple}
\begin{algorithmic}[1]
    \State $\{ (A^{(1)},B^{(1)}),\ldots,(A^{(s)},B^{(s)})\} = \textsc{SumToMaxCovering}(A,B,\eps)$
    \State $C^{(\ell)} := \textsc{MinMaxConv}(A^{(\ell)}, B^{(\ell)})$ for any $\ell \in [s]$
    \State $\tilde C[k] := \min_{\ell \in [s]} \{ C^{(\ell)}[k] \}$ for any $0 \le k < n$
    \State \Return $\tilde C$
\end{algorithmic}
\end{algorithm}
\end{proof}

\subsection{Equivalence of Approximate Min-Plus and Exact Min-Max Convolution}
\label{equiv-minconv}

We next show an equivalence of approximate \minconv and exact \minmaxconv, similarly to the equivalence for matrix products in Theorem~\ref{thm:equ-apsp}.

\begin{theorem} \label{thm:conv-equ}
  For any $c \ge 1$, if one of the following statements is true, then both are:
  \begin{itemize}
    \item $(1+\eps)$-Approximate \minconv can be solved in strongly polynomial time $\tOh(n^c/ \poly(\eps))$,
    \item exact \minmaxprod can be solved in strongly polynomial time $\tOh(n^c)$.
  \end{itemize}
\end{theorem}

In particular, any further improvement on the exponent of $n$ in Theorem~\ref{thm:minconvapxsimple} would yield an improved algorithm for \minmaxconv.

\begin{proof}
  For one direction, observe that if \minmaxconv can be solved in time $T(n)$ then the algorithm from Theorem~\ref{thm:minconvapxsimple} runs in time $\tOh(n/\eps + T(n)/\eps)$, which is $\tOh(T(n))$ for constant $\eps >0$. 
  
  For the other direction, on input $A,B$ denote the result of \minmaxconv by $C$. Set $r := \lceil 4 (1+\eps)^2 \rceil$ and consider the sequences $A',B'$ with $A'[i] := r^{A[i]}$ and $B'[j] := r^{B[j]}$. (Note that the integers $A[i],B[j]$ are in standard bit representation, so we can compute floating-point representations of $r^{A[i]},r^{B[j]}$ in constant time, essentially by writing $A[i],B[j]$ into the exponent.) Let $C'$ be the result of $(1+\eps)$-Approximate \minconv on $A',B'$.
  Then as in the proof of Theorem~\ref{thm:equ-apsp}, we obtain $r^{C[k]} \le C'[k] \le r^{C[k] + 1/2}$.
%
  Hence, we can infer the \minmaxconv $C$ of $A,B$ by setting $C[k] = \lfloor \log_r C'[k] \rfloor$ (i.e., we essentially only read the exponent of the floating-point number $C'[k]$). 
  If $(1+\eps)$-Approximate \minconv is in time $T(n)$, this yields an algorithm for \minmaxconv running in time $\tOh(n + T(n)) = \tOh(T(n))$.
%
\end{proof}

\subsection{Improved Approximation Algorithm}
\label{sec:improvedapx-minconv}

In the remainder of this section, we improve the simple approximation algorithm of Theorem~\ref{thm:minconvapxsimple}.

\begin{theorem}
    \label{apx-minconv}
    $(1+\eps)$-Approximate \minconv can be solved in strongly polynomial time $\Ot(n^{3/2}/ \sqrt{\eps})$.
\end{theorem}

We will divide the algorithm for \minconv into two parts: the first part
will handle the case when summands are \emph{close} and the second will handle
the \emph{distant} case.

\subsubsection{Approximating \minconv for Distant Summands}
\label{sec:apx-distant-minconv}

First, we will simply use Distant Covering (Corollary~\ref{distant-covering-cor}) to correctly compute \minconv for
summands that differ by at least a factor $\frac 1\eps$.

\begin{lemma}
    \label{minconv-distant}
    Given sequences $A,B \in \real^n$ and a parameter $\eps > 0$, let $C$ be the
    result of \minconv on $A,B$. In strongly polynomial $\Oh(n^{3/2}\polylog(\frac n\eps))$ time
    we can compute a sequence $\tilde{C}$ such that:
    \begin{enumerate}[align=left, font=\normalfont, label=(\roman*)]
        \item for any $k \in [n]$ we have $\tilde{C}[k] \ge C[k]$, and
        \item if there are $i+j=k$ with $C[k] = A[i] + B[j]$ and $\frac{A[i]}{B[j]} \notin
            [\frac \eps 4, \frac 4 \eps]$ then $\tilde{C}[k] \le (1+\eps)C[k]$.
    \end{enumerate}
\end{lemma}

\begin{algorithm}
    \caption{$\textsc{DistantMinConv}(A,B,\eps)$}
	\label{alg:distant-minconv}
\begin{algorithmic}[1]
    \State $\{ (A^{(1)},B^{(1)}),\ldots,(A^{(s)},B^{(s)})\} = \textsc{DistantCovering}(A,B,\eps/4)$
    \State $C^{(\ell)} := \textsc{MinMaxConv}(A^{(\ell)}, B^{(\ell)})$ for any $\ell \in [s]$
    \State $\tilde C[k] := \frac{1}{1-\eps/2} \cdot \min_{\ell \in [s]} \{ C^{(\ell)}[k] \}$ for any $0 \le k < n$
    \State \Return $\tilde C$
\end{algorithmic}
\end{algorithm}

\begin{proof}[Proof of Lemma~\ref{minconv-distant}]
    We view the sequences $A,B$ as vectors in $\real^n$, and apply Distant
    Covering (Corollary~\ref{distant-covering-cor}) with $\eps' := \frac \eps 4$. This
    yields sequences $A^{(1)},\ldots,A^{(s)},\, B^{(1)},\ldots,B^{(s)} \in
    \real^n$ with $s = \Oh(\polylog(\frac n \eps))$. We compute the \minmaxconv of
    every layer $A^{(\ell)},B^{(\ell)}$.  Then we compute the entry-wise minimum
    of the results, scale it by a factor $\frac{1}{1-2\eps'}$, and return the resulting sequence~$\tilde C$, see Algorithm~\ref{alg:distant-minconv}. 
     
    For correctness, note that the scaling factor $\frac{1}{1-2\eps'}$ removes the factor $1-2\eps'$ from the right hand side of property (i) in Corollary~\ref{distant-covering-cor}. This yields $\tilde{C}[k] \ge
    C[k]$.
    Moreover, by property (ii) of Corollary~\ref{distant-covering-cor}, for
    any indices $i+j=k$ with $\frac{A[i]}{B[j]} \notin [\eps',
    \frac{1}{\eps'}] = [\frac \eps 4, \frac 4 \eps]$ there is an $\ell$ such that $C^{(\ell)}[k] \le A[i] + B[j]$.
    Minimizing over all $\ell$ and multiplying by $\frac{1}{1-2\eps'} = \frac 1{1-\eps/2} < 1+\eps$ yields the claimed property (ii).
    
    Since Distant Covering takes time $\Oh(n \polylog(\frac n \eps))$, the running time is dominated by
    computing $s$ times \minmaxconv. Using the fastest known algorithm for \minmaxconv~\cite{first-minmaxconv}, we obtain time $\Ot(n^{3/2}s) =
    \Oh(n^{3/2} \polylog(\frac n \eps))$.
\end{proof}

\subsubsection{Approximating \minconv for Close Summands}

We now use a variant of a known scaling-based approximation scheme for \minconv to handle the close summands. 

\begin{lemma}
    \label{minconv-close}
    Given sequences $A,B \in \real^n$ and a parameter $\eps > 0$, let $C$ be the
    result of \minconv on $A,B$. In strongly polynomial $\tOh(n^{3/2}/\sqrt{\eps})$ time
    we can compute a sequence $\tilde{C}$ such that:
    \begin{enumerate}[align=left, font=\normalfont, label=(\roman*)]
        \item for any $k \in [n]$ we have $\tilde{C}[k] \ge C[k]$, and
        \item if there are $i+j=k$ with $C[k] = A[i] + B[j]$ and $\frac{A[i]}{B[j]} \in
            [\frac \eps 4, \frac 4 \eps]$ then $\tilde{C}[k] \le (1+\eps)C[k]$.
    \end{enumerate}
\end{lemma}

\begin{algorithm}
    \caption{$\textsc{Scale}(A,q,\eps)$.}
    \label{alg:conv-scale}
\begin{algorithmic}[1]
    \State $A'[i] = \begin{cases} \big\lceil \frac{4}{\eps q}\cdot A[i] \big\rceil & \text{if} \; \frac{\eps q}{16} \le A[i] \le q
         \\
		\infty & \text{otherwise} \\
	\end{cases}$
	\State \Return $A'$
\end{algorithmic}
\end{algorithm}
\begin{algorithm}
    \caption{$\textsc{CloseMinConv}(A,B,\eps)$.}
	\label{alg:approx-minconv}
\begin{algorithmic}[1]
    \State Initialize $\tilde{C}[k] := \infty$ for all $k$
    \For {$q = 1,2,4,\ldots, 2^{\ceil{\log{2W}}}$}
      \State $A' := \textsc{Scale}(A,q,\eps)$
      \State $B' := \textsc{Scale}(B,q,\eps)$
      \State $V := \textsc{ExactMinConv}(A',B')$ 
      \State $\tilde{C}[k] := \min \{\tilde{C}[k], V[k] \cdot \frac{q \eps}{4} \}$ for all $k$
    \EndFor
    \State \Return $\tilde{C}$
\end{algorithmic}
\end{algorithm}

\begin{proof}
Consider the procedures $\textsc{Scale}$ and $\textsc{CloseMinConv}$ (Algorithms~\ref{alg:conv-scale} and~\ref{alg:approx-minconv}), which are modifications of the known $\Ot(\frac{n}{\eps} \log{W})$-time approximation scheme for \minconv~\cite{partition}. The main difference to~\cite{partition} is that the entries of $A$ that are set to $\infty$ in the procedure $\textsc{Scale}$ are not only the ones that are too large (greater than $q$)  but also the ones that are too small (smaller than $\frac{\eps q}{16}$).
We claim that this algorithm proves Lemma~\ref{minconv-close}.
Correctness is based on the following rounding lemma.
\begin{lemma}[cf. Lemma B.2 in \cite{partition}]
    \label{numbers-rounding}
    For any $x,y,q,\eps \in \real$ with $x+y \ge q/2$
    and $0 < \eps < 1$ we have:
    \[
       x+y \le \big(\big\lceil \tfrac{4x}{q\eps} \big\rceil +
        \big\lceil \tfrac{4y}{q\eps} \big\rceil\big)\tfrac{q\eps}{4} \le
        (1+\eps)(x+y).
    \]
\end{lemma}
\begin{proof}
  We repeat the proof for completeness. The lower bound is immediate. For the upper bound, note that 
  \[ \big(\big\lceil \tfrac{4x}{q\eps} \big\rceil +
        \big\lceil \tfrac{4y}{q\eps} \big\rceil\big)\tfrac{q\eps}{4} \le x + y + 2 \tfrac{q\eps}{4} = x + y + \eps \cdot \tfrac q 2 \le (1+\eps)(x+y). \qedhere \]
\end{proof}

\paragraph{Correctness} 
Correctness of $\textsc{CloseMinConv}$ can now be shown similarly as in~\cite{partition}. Regarding property (i), the lower bound of Lemma~\ref{numbers-rounding} ensures that we have $\tilde C[k] \ge C[k]$ throughout the run of the algorithm. 
Regarding property (ii), for any $i+j=k$ with $C[k] = A[i] + B[j]$ there exists a precision parameter~$q$ with $q/2 \le A[i] + B[j] \le q$.
In particular, we have $A[i], B[j] \le q$ and $\max\{A[i],B[j]\} \ge \frac q4$.
If we additionally have $\frac{A[i]}{B[j]} \in [\frac \eps 4, \frac 4 \eps]$, then 
\[ \min\{A[i],B[j]\} \ge \tfrac \eps 4 \cdot \max\{A[i],B[j]\} \ge \tfrac {\eps q}{16}, \]
and thus $A'[i]$ and $B'[j]$ both are not set to $\infty$ by the procedure $\textsc{Scale}$. The upper bound of Lemma~\ref{numbers-rounding} now implies $\tilde C[k] \le (1+\eps)(A[i] + B[j]) = (1+\eps)C[k]$.


\paragraph{Running time} 
For the running time analysis, denote by $\alpha_q$ the number of entries of $A'$ that are not set to $\infty$ in iteration $q$.
Note that if an entry $A[i]$ is not set to $\infty$ in iteration $q$, then we have $\frac{\eps q}{16} \le A[i] \le q$, or, equivalently, $A[i] \le q \le \frac{16}{\eps} A[i]$. Since $q$ grows geometrically, there are $\Oh(\log \frac 1\eps)$ iterations $q$ in which entry $A[i]$ is not set to $\infty$. Hence, we obtain $\sum_q \alpha_q = \Oh(n \log \frac 1\eps)$. 
We similarly define $\beta_q$ as the number of non-$\infty$ entries of $B'$ in iteration $q$, and obtain 
\begin{align} \label{eq:summinconv}
  \sum_q \beta_q = \Oh\big(n \log \tfrac 1\eps\big).
\end{align}

We argue in the following that procedure $\textsc{CloseMinConv}$ can be implemented in such a way that the running time for iteration $q$ is $\Oh(\min\{ \alpha_q \beta_q, \frac n \eps \} \polylog( \frac n\eps))$. 

To this end, we use an event queue to be able to skip all iterations with $\alpha_q = 0$ or $\beta_q = 0$. Moreover, we can maintain all non-$\infty$ entries of $A',B'$ in time $\Oh(\alpha_q + \beta_q)$ per iteration $q$. Note that $\Oh(\alpha_q + \beta_q) \le \Oh(\min\{ \alpha_q \beta_q, \frac n \eps \})$. This yields the claimed time bound for lines 3 and 4 of procedure $\textsc{CloseMinConv}$.

For line 5 of procedure $\textsc{CloseMinConv}$, note that naively the exact \minconv of $A'$ and $B'$ can be computed in time $\Oh(\alpha_q \beta_q)$. Moreover, since $A',B'$ have entries in $\{1,\ldots,W\} \cup \{\infty\}$ for $W = \ceil{\frac 4\eps}$, by Fast Fourier Transform their \minconv can be computed in time $\tOh(W n) = \tOh(\frac n \eps)$. Using the better of the two yields the claimed time bound.

Finally, line 6 of procedure $\textsc{CloseMinConv}$ can be implemented in time $\Oh(\min\{ \alpha_q \beta_q, n \})$, since this is an upper bound on the number of non-$\infty$ entries of $V$.

Altogether, we obtain time $\Oh(\min\{ \alpha_q \beta_q, \frac n \eps \} \polylog( \frac n\eps))$ per iteration $q$. Hence, the total running time of procedure $\textsc{CloseMinConv}$ is bounded by 
\[ \sum_q \min\big\{ \alpha_q \beta_q, \tfrac n \eps \big\} \polylog(\tfrac n\eps). \]
We split this sum into the cases $\beta_q \le \lambda$ and $\beta_q > \lambda$, and note that the second case can occur at most $\Oh(\frac n\lambda \log \frac 1\eps)$ times due to (\ref{eq:summinconv}).
This yields a total running time of at most
\[ \Big( \big(\sum_q \alpha_q \lambda\big) + \tfrac n\eps \cdot \tfrac n\lambda \log \tfrac 1\eps \Big) \polylog(\tfrac n\eps) \le \big( n \lambda + \tfrac {n^2}{\eps \lambda}\big) \polylog(\tfrac n\eps). \]
Setting $\lambda := (n/\eps)^{1/2}$ yields the claimed total running time bound $\tOh(n^{3/2} / \eps^{1/2})$.
\end{proof}

\subsubsection{Proof of Theorem~\ref{apx-minconv}}

\begin{algorithm}
    \caption{$\textsc{ApxMinConv}(A,B,\eps)$.}
	\label{alg:improved-apx-miconv}
\begin{algorithmic}[1]
    \State $\tilde{C}_1 := \textsc{DistantMinConv}(A,B,\eps)$
    \State $\tilde{C}_2 := \textsc{CloseMinConv}(A,B,\eps)$
    \State $\tilde C[k] := \min \{ \tilde{C}_1[k], \tilde{C}_2[k] \}$ for any $0 \le k < n$
    \State \Return $\tilde{C}$
\end{algorithmic}
\end{algorithm}

\begin{proof}[Proof of Theorem~\ref{apx-minconv}]
    Given sequences $A,B \in \real^n$ and a parameter $\eps > 0$, we simply run
    our algorithms for approximating \minconv on distant and close summands, 
    and compute their entry-wise minimum (see Algorithm~\ref{alg:improved-apx-miconv}). Correctness as well as size and
    time bounds are immediate consequences of Lemmas~\ref{minconv-distant} and~\ref{minconv-close}. 
\end{proof}

\subsection{Applications for Tree Sparsity}

A direct consequence of the \strongapx for \minconv is an analogous \strongapx
for \sparsity. 
We will make use of the following strongly polynomial reduction that uses an approximation
algorithm for \minconv as a black-box.

\begin{theorem}[cf. Theorem 7.1 \cite{partition}, reformulation of~\cite{tree-sparsity}]
    \label{tree-sparsity-reduction}
    If $(1+\eps)$-Approximate \minconv can be solved in time $T(n,\eps)$, then $(1+\eps)$-Approximate \sparsity can be solved in time $\Oh\big(\big(n + T(n,\eps/\log^2{n})\big) \log{n}\big)$.
\end{theorem}

Note, that this reduction runs in strongly polynomial time, and the $\log{W}$-factors in the running time of~\cite{tree-sparsity,partition} come exclusively from the use of scaling for approximating
\minconv. In consequence, our \strongapx for \minconv yields a \strongapx for \sparsity.

\begin{corollary}
    \label{tree-sparsity-apx}
    $(1+\eps)$-Approximate \sparsity can be solved in strongly polynomial time
    $\Ot(\frac{n^{3/2}}{\sqrt{\eps}})$. 
\end{corollary}

\sparsity has applications in image processing, computational
biology~\cite{serang} and machine learning~\cite{tree-sparsity}. The main issue of previous
approximation schemes for this problem was that in applications the input
consists of real numbers and thus $\log{W}$-factors significantly influence the running
time~\cite{schmidt-communication}.

\bibliographystyle{abbrvnat}
\bibliography{bib}

\begin{appendices}

\section{Problem Definitions}
\label{def-problems}

For an edge-weighted (directed or undirected) graph $G$, $d_G(u,v)$ denotes the minimal total weight of any path from $u$ to $v$ in $G$.

\defproblem{Directed/Undirected All-Pairs Shortest Path (APSP)}
{An edge-weighed directed/undirected graph $G$ on $n$ nodes, with weights in $\real$}
{Compute $d_G(u,v)$ for any $u,v \in [n]$}

\defproblem{$(1+\eps)$-Approximate Directed/Undirected All-Pairs Shortest Path (APSP)}
{An edge-weighed directed/undirected graph $G$ on $n$ nodes, with weights in $\real$}
{Compute values $\tilde d(u,v)$ for any $u,v \in [n]$ such that $d_G(u,v) \le \tilde d(u,v) \le (1+\eps) d_G(u,v)$}

\defproblem{All-Pairs Bottleneck Path}
{An edge-weighed directed graph $G$ on $n$ nodes, with weights in $\real$}
{For every two nodes $u,v$ determine the maximum over all paths from $u$ to $v$ of the minimum edge weight along the path}

\defproblem{\minprod}
{Matrices $A,B \in \real^{n\times n}$}
{Compute matrix $C \in \real^{n \times n}$ with $C[i,j] = \min_{k \in [n]} (A[i,k] + B[k,j])$ for any $i,j \in [n]$}

\defproblem{\minmaxprod}
{Matrices $A,B \in \real^{n\times n}$}
{Compute matrix $C \in \real^{n \times n}$ with $C[i,j] = \min_{k \in [n]} \max\{A[i,k], B[k,j]\}$ for any $i,j \in [n]$}

\defproblem{$(1+\eps)$-Approximate \minprod}
{Matrices $A,B \in \real^{n\times n}$}
{Compute a matrix $\tilde C \in \real^{n \times n}$ with $C[i,j] \le \tilde C[i,j] \le (1+\eps) C[i,j]$ for any $i,j \in [n]$, where $C$ is the correct output of \minprod}


\defproblem{\minconv}
{Sequences $(A[i])_{i=0}^{n-1},\, (B[i])_{i=0}^{n-1} \in \real^n$} 
{Compute sequence $(C[i])_{i=0}^{n-1}$ with $C[k] = \min_{i+j=k} (A[i]+B[j])$}

\defproblem{\minmaxconv}
{Sequences $(A[i])_{i=0}^{n-1},\, (B[i])_{i=0}^{n-1} \in \real^n$}
{Compute sequence $(C[i])_{i=0}^{n-1}$ with $C[k] = \min_{i+j=k}
\max\{A[i],B[j]\}$}

\defproblem{$(1+\eps)$-Approximate \minconv}
{Sequences $(A[i])_{i=0}^{n-1},\, (B[i])_{i=0}^{n-1} \in \real^n$}
{Compute a sequence $(\tilde C[i])_{i=0}^{n-1}$ with $C[k] \le \tilde C[k] \le (1+\eps) C[k]$ for any $0 \le k < n$, where $C$ denotes the correct output of \minconv}


\defproblem{\sparsity}
{A node-weighted rooted tree $T$, with weights in $\real$, and an integer $k$}
{Find the maximal total weight of any rooted subtree of $T$ consisting of $k$ vertices}

\defproblem{Directed/Undirected Minimum Weight Triangle}
{An edge-weighed directed/undirected graph $G$ on $n$ nodes, with weights in $\real$}
{Compute the minimal total weight of any triangle in $G$}

\defproblem{$(1+\eps)$-Approximate Directed/Undirected Minimum Weight Triangle}
{An edge-weighed directed/undirected graph $G$ on $n$ nodes, with weights in $\real$}
{Compute a number $\tilde T \in [T,(1+\eps)T]$, where $T$ is the minimal total weight of any triangle in $G$}

\defproblem{Directed/Undirected Minimum Weight Cycle}
{An edge-weighed directed/undirected graph $G$ on $n$ nodes, with weights in $\real$}
{Compute the minimal total weight of any cycle in $G$}

\defproblem{$(1+\eps)$-Approximate Directed/Undirected Minimum Weight Cycle}
{An edge-weighed directed/undirected graph $G$ on $n$ nodes, with weights in $\real$}
{Compute a number $\tilde C \in [C,(1+\eps)C]$, where $C$ is the minimal total weight of any cycle in $G$}

\defproblem{Radius}
{An edge-weighed directed/undirected graph $G$ on $n$ nodes, with weights in $\real$}
{Compute $\min_{s \in V(G)} \max_{v \in V(G)} d_G(s,v)$}

\defproblem{$(1+\eps)$-Approximate Radius}
{An edge-weighed directed/undirected graph $G$ on $n$ nodes, with weights in $\real$}
{Compute a number $\tilde R \in [R,(1+\eps)R]$, where $R$ is the radius of $G$}

\defproblem{Diameter}
{An edge-weighed directed/undirected graph $G$ on $n$ nodes, with weights in $\real$}
{Compute $\max_{u,v \in V(G)} d_G(u,v)$}

\defproblem{$(1+\eps)$-Approximate Diameter}
{An edge-weighed directed/undirected graph $G$ on $n$ nodes, with weights in $\real$}
{Compute a number $\tilde D \in [D,(1+\eps)D]$, where $D$ is the diameter of $G$}

\defproblem{Median}
{An edge-weighed directed/undirected graph $G$ on $n$ nodes, with weights in $\real$}
{Compute $\min_{u \in V(G)} \sum_{v \in V(G)} d_G(u,v)$}

\defproblem{$(1+\eps)$-Approximate Median}
{An edge-weighed directed/undirected graph $G$ on $n$ nodes, with weights in $\real$}
{Compute a number $\tilde M \in [M,(1+\eps)M]$, where $M$ is the median of $G$}

\end{appendices}

\end{document}